\documentclass{ifacconf}

\usepackage{graphicx} % include this line if your document contains figures
\usepackage{natbib} % required for bibliography

\usepackage{amsmath, amsfonts}
\usepackage{algorithmic}
\usepackage{algorithm}
\usepackage{array}
\usepackage{caption}
\usepackage{subcaption}
\usepackage{textcomp}
\usepackage{stfloats}
\usepackage{url}
\usepackage{verbatim}
\usepackage{bm}
\usepackage{subfiles}
\usepackage{comment}
\usepackage{enumitem}
\newtheorem{definition}{Definition}
\newtheorem{theorem}{Theorem}
\begin{document}
\begin{frontmatter}
    \title{Stealthy Coverage Control for Human-enabled Real-Time 3D
        Reconstruction\thanksref{footnoteinfo}}

    \thanks[footnoteinfo]{This work has been supported by JSPS KAKENHI Grant Number 24K00906.}

    \author[First]{Reiji Terunuma}
    \author[First]{Yuta Nakamura}
    \author[First]{Takuma Abe}
    \author[First]{Takeshi Hatanaka}

    \address[First]{Institute of Science Tokyo (e-mail: \{terunuma,~nakamura.y,~abe.t\}@hfg.sc.e.titech.ac.jp,
        hatanaka@sc.e.titech.ac.jp).}

    \begin{abstract}                % Abstract of 50--100 words
In this paper, we propose a novel semi-autonomous image sampling strategy, called stealthy coverage control, for human-enabled 3D structure reconstruction. 
The present mission involves a fundamental problem: while the number of images required to accurately reconstruct a 3D model depends on the structural complexity of the target scene to be reconstructed, it is not realistic to assume prior knowledge of the spatially non-uniform structural complexity.
We approach this issue by leveraging human flexible reasoning and situational recognition capabilities. Specifically, we design a semi-autonomous system that leaves identification of regions that need more images and navigation of the drones to such regions to a human operator.
To this end, we first present a way to reflect the human intention in autonomous coverage control. Subsequently, in order to avoid operational conflicts between manual control and autonomous coverage control, we develop the stealthy coverage control that decouples the drone motion for efficient image sampling from navigation by the human. 
Simulation studies on a Unity/ROS2-based simulator demonstrate that the present semi-autonomous system outperforms the one without human interventions in the sense of the reconstructed model quality.

    %In this paper, we address coordinated image sampling with multi-drone systems for reconstructing a scene in three dimensions. %
    %To this end, we start with presenting a 3D coverage control strategy for efficiently sampling images from rich viewing angles. %
    %The number of images required for reconstructing an accurate 3D model depends on the actual structure of the scene which is a priori unknown. %
    %A promising approach to the issue is to leave the model evaluation in real
    %time and navigation of the drones to the areas that need more images to a
    %human operator. %
    %We thus present a novel control architecture coordinating a human and
    %multiple drones. %
    %Now, coverage control for efficient image sampling interferes the human operation. %
    %To address the issue, we present a novel stealthy coverage control such
    %that the drone actions driven by coverage control are ensured to be
    %invisible from the human operator. %
    %We finally demonstrate the present control strategy through human-in-the-loop
    %simulation.
\end{abstract} %

    \begin{keyword}
        Coverage control, Human-in-the-loop, Stealthy control, 3D
        reconstruction, CPHS
    \end{keyword}
\end{frontmatter}
%===============================================================================

%%============================================================================%%
\section{Introduction}
\label{chp:introduction}
%%============================================================================%%
Reconstructing three-dimensional (3D) 
structure from a collection of images has gained increasing interests and is expected to be
a key solution to various fields such as precision agriculture (\cite{Edmonds:2021}) 
and infrastructure inspection (\cite{Seraj:2020}).
The high-quality and accurate 3D structural models enable one
to observe and analyze the current state of the environment in greater
detail than ever before.

Autonomous control of drones is a promising approach to efficient image sampling for the 3D reconstruction. 
Coordinating multiple drones is also expected to further enhance the sampling efficiency, as compared to the single drone operation (\cite{Torres:2016}).
%employing multiple drones enables more efficient coverage and data sampling. %
Early approaches rely on precomputed flight paths (\cite{Xiao:2021}).
However, they are inherently inflexible against various disturbances and 
the change of the number of drones in operation.
Moreover, they also lack a systematic design for camera rotation control, despite the expectation that variable camera orientations would enhance the diversity of the viewing angles.
%as they cannot perform operations such as re-capturing specific areas or expanding 
%the coverage during flight.
A promising approach to ensuring flexibility is to employ coverage control 
%To address this inflexibility, we focus on coverage control 
(\cite{Cortes:2005}, \cite{Schwager:2011}, \cite{Palacios:2016}, \cite{Dan:2021}),
which deploys mobile robots over a mission space so that the environmental data are sampled efficiently.
%is an algorithmic framework for efficiently allocating multiple robots by dynamically determining control inputs based on the environment.
% Coverage control refers to algorithms that efficiently allocate multiple robots in space. %
%As an application, several studies (e.g. \cite{Palacios:2016}, \cite{Dan:2021})
%have proposed dynamically varying the importance of each location and
\cite{Shimizu:2022} proposed a coverage controller specialized to the image sampling for 3D reconstruction, 
where sampling images from diverse viewing angles is enforced by utilizing 
%control barrier functions
constraint-based control
(\citet{Magnus:2021}). 
%to enable persistent monitoring.
The control strategy was further extended to the case with camera rotational control (\cite{Lu:2024}), 
enabling systematic determination of camera orientations.
However, these solutions lack another kind of flexibility.
%To adopt coverage control framework for 3D reconstruction, \cite{Shimizu:2022} proposed multi-angle 
%coverage control for farmland mapping, which ensures accurate reconstruction by capturing the environment
%from various angles.
Namely, the number of images required to accurately reconstruct the 3D structure 
is highly dependent on the structural complexity of the scene, 
and this complexity is not uniform across the scene.
Since the distribution of the complexity is hardly assumed as prior knowledge, 
\cite{Shimizu:2022} and \cite{Lu:2024} required drones to uniformly sample images  over the scene.
To address the issue, \cite{Hanif:2025} presented a coordinated
image sampling framework that utilizes feedback from mesh changes for 
the structure model that evolves in real time.
However, detecting mesh changes is computationally expensive,
making the approach impractical for large-scale environment.

% However, even with such methods, the reconstructed models obtained through Structure from Motion (SfM) and Multi-View Stereo (MVS)
% still exhibits incomplete regions (\cite{Suenaga:2022}). %
% This is likely due to the fact that some images are not used in the reconstruction
% process, as they may be blurred or affected by light reflections during capture. %
% It is difficult to detect such deficiencies in advance. %
% Moreover, due to the long processing time required for SfM and MVS, it is also
% challenging to identify reconstruction flaws during drone operation. %
% Even if reconstruction could be completed in real time, evaluating its quality would
% still be difficult for a computer that does not perceive the real environment. %
% % On the other hand, humans are capable of forming mental models of 3D structures
% % from observed images or visual input. %
% In contrast, humans are able to interpret 3D structures from visual information internally. %
% By comparing the reconstructed models with their internal mental images, humans can
% point out inconsistencies or deficiencies. %
% Therefore, if reconstruction results could be obtained in real time, it would be
% possible for humans to assess areas or directions with insufficient accuracy and
% guide the drones to recapture those regions accordingly.

Another promising approach is to assume human interventions. Due to the high capability of the humans in flexible reasoning and situational recognition, 
a human is expected to identify the region having low reconstruction quality from the structural model given in real time.
Human-multiple-robot collaborations have been in-depth studied, as summarized in \cite{Hatanaka:2023}.
For example, stable navigation of multiple robots was investigated in \cite{Lee:2005}, \cite{Franchi:2012a}, \cite{Franchi:2012b}, \cite{Atman:2019} and \cite{Hatanaka:2024}, 
which would be useful for leading drones to a region with low model quality identified by the human operator.
However, these papers employed formation control or motion synchronization laws as inter-robot distributed autonomous control, 
and it cannot be directly applied to the scenario of coordinated image sampling.
The only exception was reported by \cite{Diaz:2017}, 
where the authors investigate coordination between a human and multiple robots running coverage control. 
However, the solution is restricted to 2D coverage control and is not compatible with the present scenario.
%Cooperative human-robotic-swarm systems have been explored for coverage control, 
%where a human designates important locations for the swarm to cover (\cite{Diaz:2017}),
%but this research is limited to 2D and lacks stability analysis of the 
%human-in-the-loop system.
%While passivity-based stability guarantees exist for human guidance of multiple robot systems
%(\cite{Hokayem:2006, Nuno:2011}), these often focus on consensus-based motion where the human 
%controls the average position (\cite{Atman:2019, Hatanaka:2024}). 
%The stability of combining human interaction with coverage control remains an open question.
%A direct summation of the human's navigation input and the coverage control input is problematic,
%as these inputs may conflict and compromise system stability.
%Therefore, a method is needed to prevent the coverage control input from affecting the average state of the drones.

% In manipulator control, methods that exploit redundant degrees of freedom 
% have been proposed to perform a given task while preserving certain states
% (\cite{Ott:2008}, \cite{Music:2017}).
Redundancy has been leveraged in robotics to support multi-objective behaviors.
Prior work demonstrated compliant nullspace control in manipulation (\cite{Ott:2008}), 
subtask decoupling in human-robot teleoperation (\cite{Music:2017}), 
and behavioral coordination in autonomous robots (\cite{Antonelli:2008}).
Inspired by these redundancy-based frameworks, this paper proposes a semi-autonomous cooperative coverage
control, called stealthy coverage control.
In the stealthy coverage control, the 3D model is reconstructed from the sampled images in real time,
and visually fed back to the human operator.
The human operator is then responsible for evaluating the
quality of the reconstructed 3D model and navigating multiple
drones, while efficient image sampling is performed by coverage control. %
Leveraging redundancy allows the system to ensure stability while maintaining intuitive operability for the human operator, 
enabling cooperative image sampling without compromising control performance.
% This provides both system stability and intuitive operability for the human operator, 
% enabling cooperative image sampling without compromising control performance.
% The sampled images are processed in real time
%  by NeuralRecon (\cite{Sun:2021}) algorithm,
% and reconstructed 3D model is visually fed back to human operator. % %
We finally demonstrate the effectiveness of the control strategy through human-in-the-loop simulation.

%
%%============================================================================%%
\section{Scenario and System Architecture}
\label{chp:model}
%%============================================================================%%
\begin{figure}
    \begin{center}
        \includegraphics[width=6.8cm]{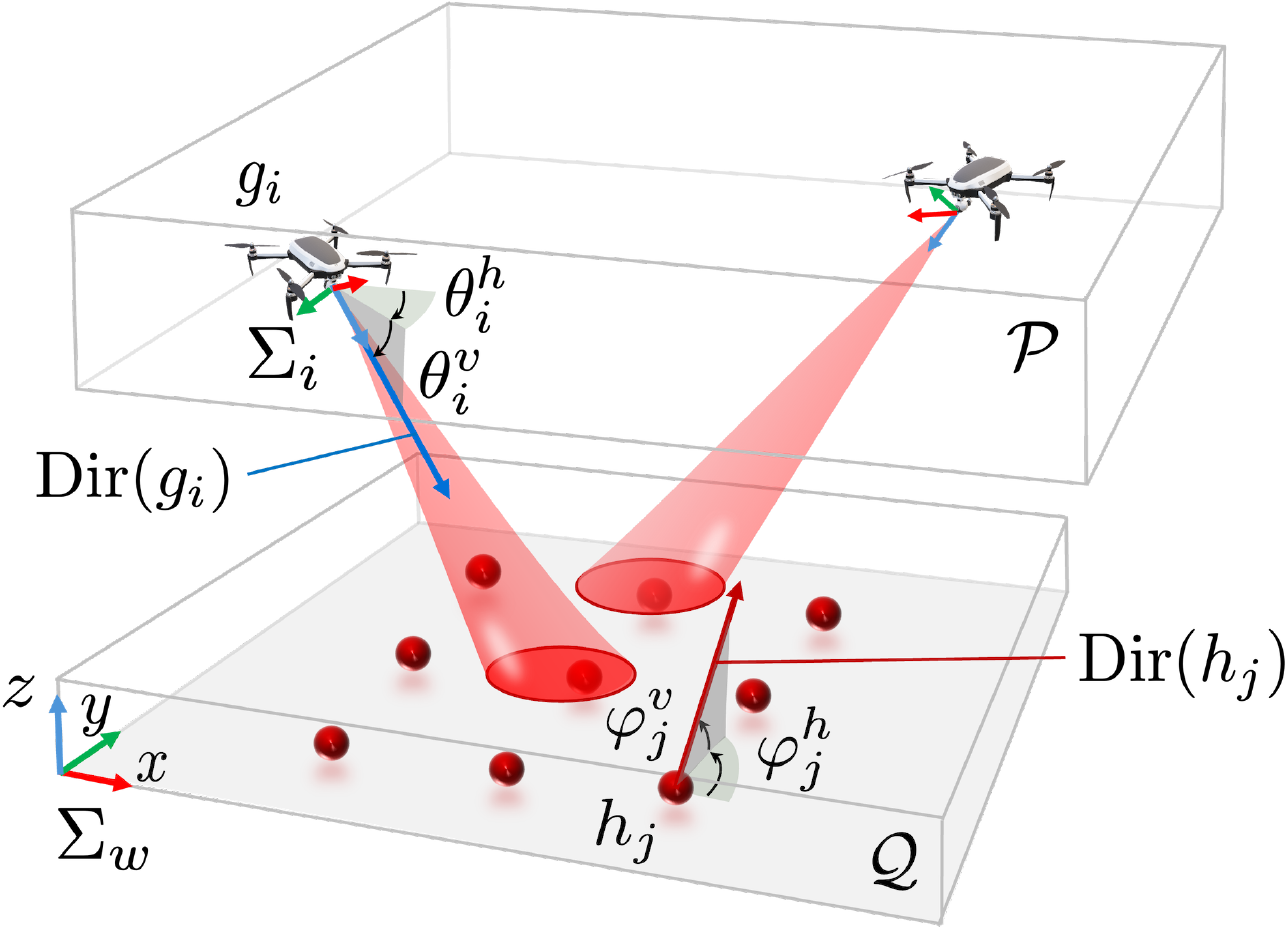}
        \caption{Illustration of the target scenario, where multiple drones sample aerial
        images of the set $\mathcal{Q}$ to reconstruct a 3D model of a scene.}
        \label{fig:scenario}
    \end{center}
\end{figure}
%%============================================================================%%
\subsection{Environment and Drone Dynamics}
%%============================================================================%%
Let us consider a cooperative image sampling task for 3D reconstruction, performed by a team of a human operator
and $n$ drones, within a bounded Euclidean space as illustrated in Fig.~\ref{fig:scenario}. %
We introduce two right-handed coordinate
systems: a stationary world frame $\Sigma_w$ and a body frame $\Sigma_i$
for each drone $i\in\mathcal{I}=\{1, 2, \ldots, n\}$. %
The world frame $\Sigma_w$ is defined such that
its $z$-axis points upward, opposite to gravity,
representing the drone flight altitude.
The body frame $\Sigma_i$ is defined so that its $z$-axis aligns with the camera's optical axis. %

The drones are assumed to operate within a compact subset $\mathcal{P}\subset\mathbb{R}^{3}$. %
The position of each drone, defined as the origin of $\Sigma_i$, in the world frame $\Sigma_w$ is denoted by
$p_i=[x_{i}\ y_{i}\ z_{i}]^{\top}\in\mathcal{P}$. %
% Each drone is equipped with a gimbal-mounted camera.
% Each drone is equipped with a camera, allowing control of its
% yaw angle $\theta^h_i\in\Theta^h=[-\pi,\pi]$ and pitch angle $\theta^v_i\in\Theta^v=(0,\pi/2]$. %
The yaw angle $\theta^h_i\in\Theta^h=[-\pi,\pi]$ is controlled through the drone's body rotation,
while the camera's pitch angle $\theta^v_i\in\Theta^v=(0,\pi/2]$ is adjusted via a gimbal. %
The state of each drone $i$ is then defined by the collection of $p_i$ and $\theta_i=[\theta^h_i\ \theta^v_i]^\top$ 
as $g_i=[p_i^\top\ \theta_i^\top]^{\top}$.
% We model the dynamics of each drone as a single-integrator system, where the state
% $p_i$ is directly controlled by its velocity input. %
The motion of each drone $i$ is assumed to be governed by the following dynamics:
\begin{align}
    \label{eq:kinematics} %
    \dot{g}_{i}=\begin{bmatrix}
        \dot{p}_i \\[1pt] \dot{\theta}_i
    \end{bmatrix}=
    \begin{bmatrix}
        {u_i^p} \\[1pt] {u_i^\theta}
    \end{bmatrix}=:u_{i},
\end{align}
where $u_{i}\in\mathcal{U}\subseteq\mathbb{R}^5$ is the collection of the
translational and angular velocity inputs to be designed,
and $\mathcal{U}$ is the admissible input set. %

%Hereafter, for any input $u_*$, a superscript $p$ denotes the translational part, while a superscript $\theta$ denotes the angular part.

High-quality 3D model is known to be obtained by observing points in the target set $\mathcal{Q}\subset\mathbb{R}^{3}$
from rich viewing angles.
To formalize this requirement, we define a 5D space
$\mathcal{H}_c=\mathcal{Q}\times\Phi^h\times\Phi^v$, %
where the sets $\Phi^h$ and $\Phi^v$ characterize the viewing angles from which the drones should sample images.
The set $\mathcal{H}_c$ is then discretized into $m$ equally sized cells. %
Let $\mathcal{M}= \{1, 2, \dots, m\}$ denote the set of cell indices, and
$\mathcal{H}$ denote the set of the five quantities associated with all cells. %
The quantities for $j \in \mathcal{M}$ are represented as $h_{j}=[q_j^\top\ \varphi_j^\top]^{\top}\in \mathcal{H}$,
where $q_j\in\mathcal{Q}$ denotes the position of the center of the cell $j$ in $\Sigma_w$,
and $\varphi_j=[\varphi^h_j\ \varphi^v_j]^\top\in\Phi^h\times\Phi^v$ denotes the target horizontal and vertical angles (Fig.~\ref{fig:scenario}).
% $\varphi^h_j\in\Theta^h$ and $\varphi^v_j\in\Theta^v$. %

%%============================================================================%%
\begin{figure}
    \begin{center}
        \includegraphics[width=8.8cm]{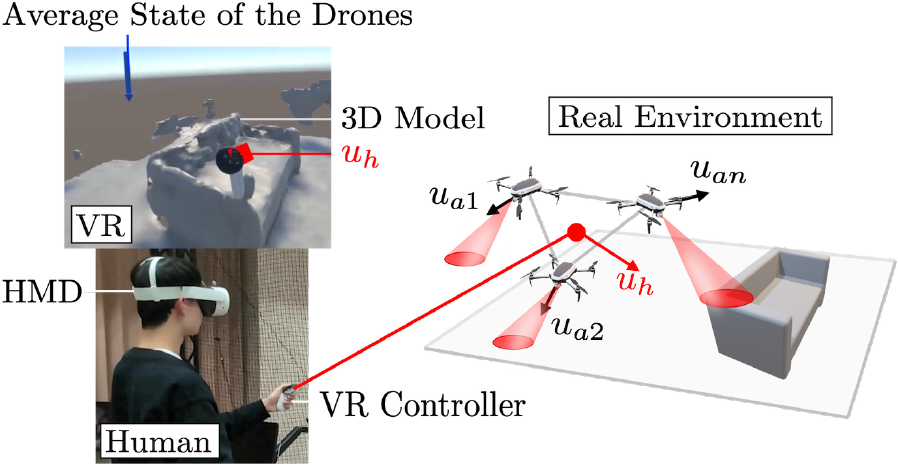}
        % \vspace{0.01mm}
        \caption{Illustration of the intended human-robots interactions.
        The human wears a head-mounted display (HMD),
        visually perceives the reconstructed 3D model and the average state of the drones in the virtual space, 
        and determines velocity commands to control the virtual average drone.}
        \label{fig:system_overview}
    \end{center}
\end{figure}

To make directional comparisons straightforward, 
we define a mapping $\mathrm{Dir}(\cdot)$ as follows:
% let $r_i\in\mathbb{R}^3$ denote the camera's viewing direction
% and $s_j\in\mathbb{R}^3$ denote the target direction of each cell:
\begin{align*}
    \mathrm{Dir}(g_i)  & = [\cos\theta^{v}_{i}\cos\theta^{h}_{i}\ \cos\theta^{v}_{i}\sin\theta^{h}_{i}\ \sin\theta^{v}_{i}]^{\top},      \\
    \mathrm{Dir}(h_j) & = [\cos\varphi^{v}_{j}\cos\varphi^{h}_{j}\ \cos\varphi^{v}_{j}\sin\varphi^{h}_{j}\ \sin\varphi^{v}_{j}]^{\top},
\end{align*}
where $\mathrm{Dir}(g_i)$ and $\mathrm{Dir}(h_j)$ represent the direction vectors
corresponding to the camera orientation of drone $i$ and the target viewing direction of cell $j$, respectively (Fig.~\ref{fig:scenario}).
% \begin{itemize}
%     \item $\mathrm{Pos}(\cdot)$: Extracts the 3D position.
%           \begin{align*}
%               \mathrm{Pos}(p_i) = [x_i\ y_i\ z_i]^{\top},\ \mathrm{Pos}(q_j) = [x_j\ y_j\ z_j]^{\top}.
%           \end{align*}
%     \item $\mathrm{Ang}(\cdot)$: Extracts the angles.
%           \begin{align*}
%               \mathrm{Ang}(p_i) = [\theta^{h}_i\ \theta^{v}_i]^{\top},\ \mathrm{Ang}(q_j) = [\theta^{h}_j\ \theta^{v}_j]^{\top}.
%           \end{align*}
%     \item $\mathrm{Dir}(\cdot)$: Converts angles to a direction vector.
%           \begin{align*}
%               \mathrm{Dir}(p_i) & = [\cos\theta^{v}_i\cos\theta^{h}_i\ \cos\theta^{v}_i\sin\theta^{h}_i\ \sin\theta^{v}_i]^{\top}, \\
%               \mathrm{Dir}(q_j) & = [\cos\theta^{v}_j\cos\theta^{h}_j\ \cos\theta^{v}_j\sin\theta^{h}_j\ \sin\theta^{v}_j]^{\top}.
%           \end{align*}
% \end{itemize}
%%============================================================================%%
\subsection{Human-enabled System Architecture}

We suppose that the drones transmit image data to a central computer, 
where the 3D structure is reconstructed in real time through a subroutine such as NeuralRecon (\cite{Sun:2021}, \cite{Hanif:2025}).
Now, the number of images required for accurate reconstruction depends on 
complexity of the target structure, which varies depending on the location within the environment. 
It is however challenging to know such spatially heterogeneous structural complexity in advance. 
Furthermore, in the absence of ground-truth structural information, 
automatically identifying regions that require additional image data 
is not always straightforward. 
In this paper, we leave the identification of regions with imperfect reconstruction quality 
and navigation of the drones to a human operator.
An overview of the system is illustrated in Fig.~\ref{fig:system_overview}. %

In the present semi-autonomous system, both a human operator and autonomous control concurrently 
contribute to the operation of the drones.
The human operator visually perceives the reconstructed model and the average state of the drones\footnote{
The architecture that the operator feeds back the average information of the drones is inspired by \cite{Atman:2019}, \cite{Hatanaka:2024}, where stability of a human-in-the-loop system was rigorously proved 
under the assumption of human passivity or passivity shortage in the translational control. 
The stability analysis can be extended to the dynamics including camera orientations, 
but this issue is beyond the scope of this specific paper and will be presented in a separate paper.}, 
and then provides a velocity command $u_{h}\in \mathbb{R}^{5}$. %
% then provide a velocity command $u_{h}=[{u_h^p}^\top\ {u_h^\theta}^\top]^\top\in \mathbb{R}^{5}$. %
The autonomous controller for each drone also computes an autonomous input
$u_{ai}\in \mathbb{R}^{5}$ for efficiently sampling images.
% $u_{ai}=[{u_{ai}^p}^\top\ {u_{ai}^\theta}^\top]^\top\in \mathbb{R}^{5}$ to realize efficient image sampling.
The velocity input $u_i$ for drone $i$ is given as the sum of the human command $u_h$
and the autonomous control input $u_{ai}$ as follows:
\begin{align}
    \label{eq:actual_control_input} %
    u_{i}=u_{h}+u_{ai}.
\end{align}
%To facilitate the human operator's decision making, the system provides
%visual feedback of the average state of the drones and 3D mesh reconstructed 
%from images captured by the drones.
% The first is the average state of all drones $\bar{g}$, which is 
% computed by averaging the individual states $g_{1}, g_{2}, \dots, g_{n}$.
% The second is a 3D mesh reconstructed in real-time, generated from the drone's
% camera images.
The average state $\bar{g}\in\mathbb{R}^6$
of the drones to be fed back to the operator is defined as follows:
% \begin{align*}
%     \bar{p} & = \begin{bmatrix}
%                     \overline{\mathrm{Pos}(p)}^{\top} & \overline{\mathrm{Ang}(p)}^{\top}
%                 \end{bmatrix}^{\top},
% \end{align*}
% where
% \begin{align}
%     \overline{\mathrm{Pos}(p)} & =\dfrac{1}{n}\sum_{i=1}^{n}\mathrm{Pos}(p_{i}),\label{eq:definition-Pos}                                  \\
%     \overline{\mathrm{Dir}(p)} & =\dfrac{\sum_{i=1}^{n} \mathrm{Dir}(p_i)}{\|\sum_{i=1}^{n} \mathrm{Dir}(p_i)\|}.\label{eq:definition-Dir}
% \end{align}
\begin{align*}
    \bar{g} & = \begin{bmatrix}
                    \bar{p}^{\top} & \overline{\mathrm{Dir}(g)}^{\top}
                \end{bmatrix}^{\top},
\end{align*}
where
\begin{align*}
    \bar{p}  =\dfrac{1}{n}\sum_{i\in\mathcal{I}}p_{i},\quad
    \overline{\mathrm{Dir}(g)}  =\dfrac{\sum_{i\in\mathcal{I}} \mathrm{Dir}(g_i)}{\|\sum_{i\in\mathcal{I}} \mathrm{Dir}(g_i)\|},
    % \bar{r}  =\dfrac{\sum_{i=1}^{n} r_i}{\|\sum_{i=1}^{n} r_i\|}.
\end{align*}
and $g\in\mathbb{R}^{5n} $ is the collective state of all drones.

In this paper, we address the following two issues:
\begin{enumerate}[label=\roman*)]
    \item How to reflect human commands to autonomous image sampling control.
    \item How to avoid conflicts between manual control $u_h$ and autonomous control $u_{ai}$ in order to preserve human operability in the shared control system.
\end{enumerate}
To address these issues, we take a two-step approach.
In the next section, we first present a novel coverage control strategy that incorporates human interventions.
Subsequently, we present a controller that resolves the potential conflicts between manual control and autonomous control. %
%%============================================================================%%
\section{Coverage Control}
\label{chp:converage_control}
%%============================================================================%%
% \subsection{Sensing Performance}
%%============================================================================%%
\begin{figure}
\centering
        \includegraphics[width=5cm]{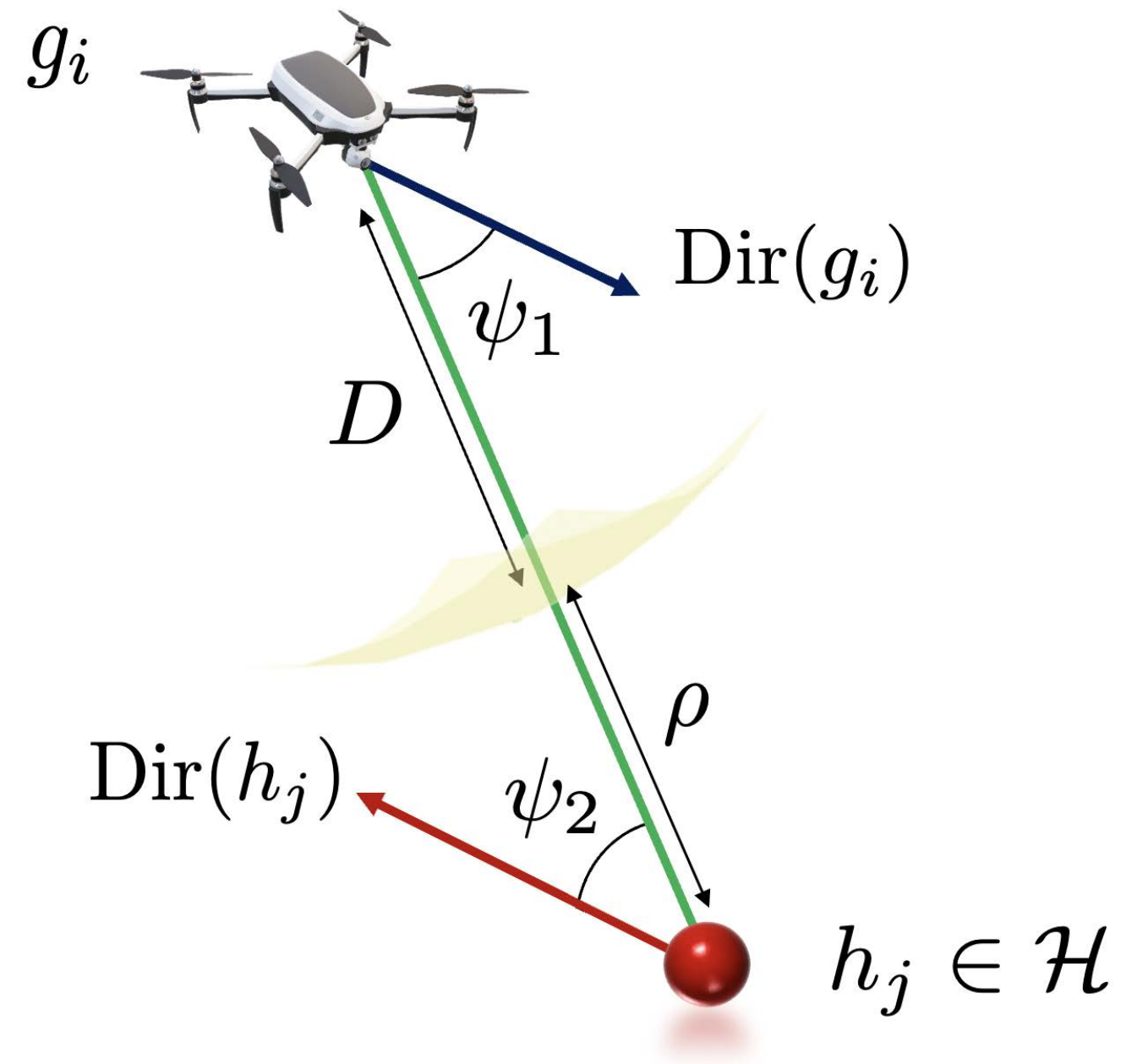}
        \caption{Geometric relation between a drone with state $g_{i}$ and an observation point
        $h_j$.}
        \label{fig:sample_point}
\end{figure}
%To evaluate the seising performance of a drone camera $i$ with state $g_{i}$ associated with the observation point $h_j\in\mathcal{H}$, %

In this section, we present a coverage control strategy that incorporates the human interventions.

Let us first define a sensing performance function that evaluates the quality of the acquired sensing data about each observation point $h_j$ acquired by drone $i$ with state $g_i$. %
%This function is based on the relative position and orientation between the drone and the observation point,
%as illustrated in Fig.~\ref{fig:sample_point}. %
%We begin by reviewing the formulation of the sensing performance function 
%presented in 
To this end, we begin by reviewing the formulation in \cite{Hanif:2025}.
Define $\psi_{1}$, $\psi_{2}$, and $\rho$ as follows:
% \begin{align*}
%     \label{eq:weight_function} %
%     \psi_{1}(p_{i}, q_j) & =\arccos{\bigg(\mathrm{Dir}(p_{i})\cdot\dfrac{\mathrm{Pos}(q_j)-\mathrm{Pos}(p_i)}{\|\mathrm{Pos}(q_j)-\mathrm{Pos}(p_i)\|}\bigg)}, \\
%     \psi_{2}(p_{i}, q_j) & =\arccos{\bigg(-\mathrm{Dir}(q_j)\cdot\dfrac{\mathrm{Pos}(q_j)-\mathrm{Pos}(p_i)}{\|\mathrm{Pos}(q_j)-\mathrm{Pos}(p_i)\|}\bigg)},  \\
%     \rho(p_{i},q_j)         & = \|\mathrm{Pos}(q_j)-\mathrm{Pos}(p_{i})\|-D,
% \end{align*}
\begin{align*}
    \label{eq:weight_function} %
    \psi_{1}(g_{i}, h_j) & =\arccos{\bigg(\mathrm{Dir}(g_i)\cdot\dfrac{q_j-p_i}{\|q_j-p_i\|}\bigg)}, \\
    \psi_{2}(g_{i}, h_j) & =\arccos{\bigg(-\mathrm{Dir}(h_j)\cdot\dfrac{q_j-p_i}{\|q_j-p_i\|}\bigg)},  \\
    \rho(g_{i},h_j)         & = \|q_j-p_i\|-D,
\end{align*}
where $D\in\mathbb{R}$ is an appropriate distance providing sufficient image resolution.
The physical meaning of these quantities is illustrated in Fig.~\ref{fig:sample_point},
and see \cite{Hanif:2025} for more details.
% and lies within the depth of field. %
Using these quantities, we define the sensing performance function $f_1(g_{i}, h_j) \in [0,1]$ as:
\begin{align}
     & f_1(g_{i},h_j) =\notag                                                                                                            \\
     & \exp\bigg( \frac{-(1-\cos \psi_1)^2}{2\sigma_1^2}+\frac{-(1-\cos \psi_2)^2}{2\sigma_2^2}+\frac{-\rho^2}{2\sigma_3^2}\bigg),
\end{align}
where a value of 1 signifies an optimal observation, and 0 indicates no observation. 
The parameters $\sigma_1$, $\sigma_2$, and $\sigma_3\in\mathbb{R}$ are constants that adjust
the sensitivity to each geometric component of the performance measure.

\cite{Hanif:2025} assign 
an importance index $\phi_j\in[0,1]$ to each observation point $q_j$, and 
updates the indices so that its decay rate follows the value of the performance function $f_1(g_{i},h_j)$. 
Accordingly, the drones are directed to unobserved regions from well-observed ones.
However, this procedure does not achieve guidance to areas where more image data should be sampled, 
as determined by the operator. 
Since the operator controls the virtual average drone, 
we assume that the observation points where the average drone achieves high sensing performances to be those where the human has decided to sample more images.
%does not prioritize areas monitored by the average human-operated drones.
%To reflect the human preference, 
According to this hypothesis, we define another sensing performance function 
$f_2(\bar{g}, h_j) \in [0,1]$ that evaluates 
the sensing performance of $h_j$ by the virtual drone having the average state $\bar g$ operated by the human as: %
\begin{align}
    \bar\psi_{1}(\bar{g}, h_j)     & =\arccos{\bigg(\overline{\mathrm{Dir}(g)}\cdot\dfrac{q_j-\bar{p}}{\|q_j-\bar{p}\|}\bigg)},\notag \\
    \bar\psi_{2}(\bar{g}, h_j)     & =\arccos{\bigg(-\mathrm{Dir}(h_j)\cdot\dfrac{q_j-\bar{p}}{\|q_j-\bar{p}\|}\bigg)},\notag\\
    f_2(\bar{g}, h_j)&=\exp\bigg( \frac{-(1-\cos \bar\psi_1)^2}{2\bar\sigma_1^2}+\frac{-(1-\cos \bar\psi_2)^2}{2\bar\sigma_2^2}\bigg),
\end{align}
where $\bar\sigma_1$ and $\bar\sigma_2\in\mathbb{R}$ are constants defined, similarly to 
$\sigma_1$ and $\sigma_2$.
% \begin{align}
%     \bar\psi_{1}(p, q_j)     & =\arccos{\bigg(\overline{\mathrm{Dir}(p)}\cdot\dfrac{\mathrm{Pos}(q_j)-\overline{\mathrm{Pos}(p)}}{\|\mathrm{Pos}(q_j)-\overline{\mathrm{Pos}(p)}\|}\bigg)},\notag \\
%     \bar\psi_{2}(p, q_j)     & =\arccos{\bigg(-\mathrm{Dir}(q_j)\cdot\dfrac{\mathrm{Pos}(q_j)-\overline{\mathrm{Pos}(p)}}{\|\mathrm{Pos}(q_j)-\overline{\mathrm{Pos}(p)}\|}\bigg)},\notag\\
%     h_2(p, q_j)&=\exp\bigg( \frac{-(1-\cos \bar\psi_1)^2}{2\bar\sigma_1^2}+\frac{-(1-\cos \bar\psi_2)^2}{2\bar\sigma_2^2}\bigg).
% \end{align}

Based on the functions $f_1(g_{i},h_j)$ and $f_2(\bar{g}, h_j)$, we update the importance indices by
\begin{align}
    \dot{\phi}_j=\delta k_{\max}(\bar{g},g_i,h_j)^2\Big(\dfrac{k_{\max}(\bar{g},g_i,h_j)+1}{2}-\phi_{j}\Big),
    \label{eq:local_importance_index} %
\end{align}
where
\begin{align}
    k(\bar{g},g_i,h_j)&=f_2(\bar{g},h_j)-f_1(g_i,h_j),\notag\\
    k_{\max}(\bar{g},g_i,h_j)&=\max_{i\in\mathcal{I}}\ k(\bar{g},g_i,h_j),\notag
\end{align}
% \begin{align}
%     \label{eq:local_importance_index} %
%     k(p_i,q_j)&=h_2(p,q_j)-h_1(p_i,q_j)\notag\\
%     k_{\max}(p_i,q_j)&=\max_{i\in\mathcal{I}}\ k(p_i,q_j)\notag\\
%     \dot{\phi}_j&=\delta k_{\max}(p_i,q_j)^2\Big(\dfrac{k_{\max}(p_i,q_j)+1}{2}-\phi_{j}\Big).
% \end{align}
and $\delta\in\mathbb{R}$ is a positive gain parameter. 
%that specifies the sensitivity of $\phi_j$ to the sensing performance.
The importance of $h_j$ increases when it is well observed by the human-operated virtual drone, 
and decreases once it has been observed with sufficient accuracy.
%%============================================================================%%
% \subsection{Objective Function}
%%============================================================================%%

%The drones are required to maximize the average sensing performance across all cells.
%In other words, t
The objective function to be minimized by the drones is formulated as follows:
\begin{align*}
    J=\sum_{j\in\mathcal{M}}\dfrac{\phi_{j}}{m},
\end{align*}
where $m$ is the cardinality of the index set $\mathcal{M}$.
Let us next partition the set of cell indices $\mathcal{M}$.
% We define the set $\mathcal{V}_{i}(g)$ for drone $i$
% as the set of indices for which that drone achieves the highest sensing performance.
% This creates a Voronoi-like partition of the set of observation points, where each cell is 
% assigned to the drone best suited to observe it, as defined by the following equation:
Define the set $\mathcal{V}_{i}(g)$ as %which is a Voronoi-like partition of the set $\mathcal{M}$ as
\begin{align*}
    \mathcal{V}_{i}(g)=\left\{j\subset\mathcal{M}\mid f_1(g_{i},h_{j})\geq f_1(g_{k},h_{j}),\ \forall{k}\subset\mathcal{I}\right\}.
\end{align*}
Since the sets $\mathcal{V}_{i}(g)$ are defined to be mutually exclusive,
the global objective function $J$ can be expressed as a sum of the local objective
functions $J_{i}$ for each drone as:
\begin{align*}
    % \label{eq:objective_function}
    J=\sum_{i\in\mathcal{I}}J_{i},\quad J_i=\sum_{j\in\mathcal{V}_i(g)}\dfrac{\phi_{j}}{m}.
\end{align*}
%%============================================================================%%
% \subsection{QP-based Controller Design}
%%============================================================================%%

Let us now consider the case where the human command is $u_{h}=0$, 
and hence the velocity input $u_i$
is determined solely by the coverage control input $u_{ci}\in \mathbb{R}^{5}$
% is determined solely by the coverage control input $u_{ci}=[{u_{ci}^p}^\top\ {u_{ci}^\theta}^\top]^\top\in \mathbb{R}^{5}$
as the autonomous control input $u_{ai}$. %
The drone dynamics are then given by $\dot g_{i}=u_{ci}$.
Similarly to \cite{Hanif:2025}, we enforce the drones to meet the inequality constraint 
$\dot J \leq -\gamma$, where $\gamma > 0$ specifies the minimal decay rate of the function $J$. 
%The control objective is %for each drone to find a minimal-norm input $u_{ci}$ that
%collectively 
%to ensure that the global objective function $J$ decreases at a rate of at least $\gamma$.
%To enforce this performance constraint, we utilize the constraint-based control approach.
% as in (\cite{Hanif:2025}).
To this end, we define $b_{i,I}=I_i-|\mathcal{V}_i(g)|\gamma/m$ 
where $I_i = -\dot J_i$.
If $b_{i,I}\geq 0$ is satisfied by every drone $i$, the global objective $\dot J \leq -\gamma$ holds.
%corresponding to $\dot{J}\leq-\gamma$,
%as the constraint 
Based on the notion of so-called constraint-based control (\cite{Magnus:2021}), we design the control input that meets $b_{i,I}\geq 0$
%The control input $u_{ci}^{*}$ for each drone $i$ is determined 
by solving the following 
quadratic program:
% Each drone $i$ selects $u_{ci}\in\mathcal{U}$ to minimize $J$, while maintaining $\dot{J}\leq-\gamma$,
% where $\gamma - \dot{J}$, explicitly determined by $h_1(p_i,q_j)$, $h_2(p,q_j)$, and $\phi_j$,
% can be treated as a control barrier function as shown in (\cite{Hanif:2025}).
\begin{align}
    u_{ci}^{*}   & =\underset{u_{ci}}{\text{arg min}}\|u_{ci}\|^{2}\notag                                                                             \\
    \text{s.t.}\ &\sum_{j\in\mathcal{V}_i(g)}\bigg(\dfrac{\partial b_{i,I}}{\partial{\phi_j}}\bigg)\dot\phi_j
    +\bigg(\dfrac{\partial b_{i,I}}{\partial g_i}\bigg)^{\top}u_{ci}+\alpha(b_{i,I})\geq 0,
    \label{eq:coverage_control}
\end{align}
where $\alpha(\cdot)$ is a locally Lipschiz extended class-$\mathcal{K}$ function (\cite{Ames:2017}).
 %
%%============================================================================%%
\section{Stealthy Coverage Control}
\label{chp:stealt_control}
%%============================================================================%%
\begin{figure*}[t]
    \begin{center}
        \includegraphics[width=12cm]{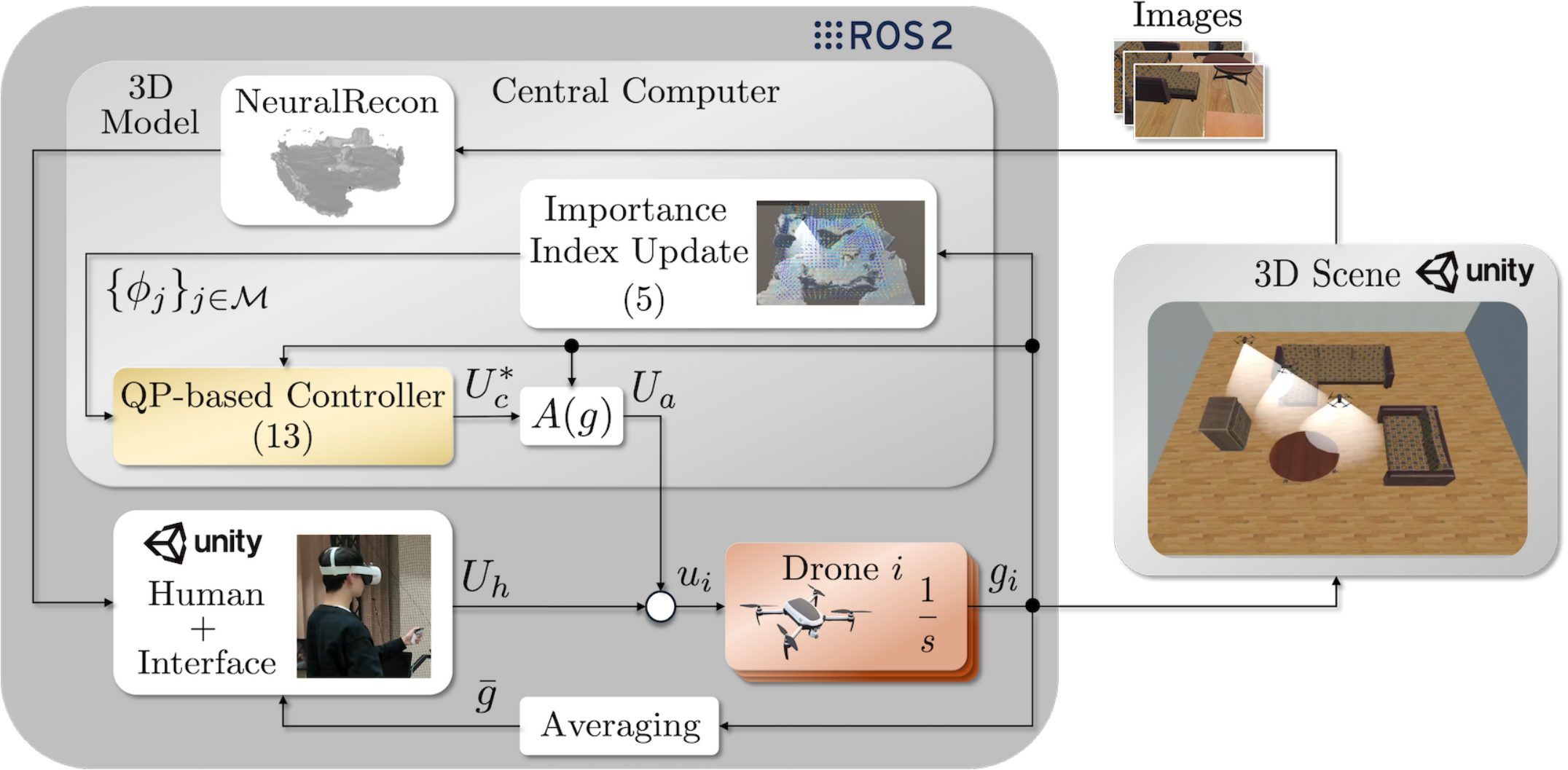}
        \caption{System architecture of the present semi-autonomous image sampling control.
        The left block shows stealthy coverage control, real-time structure reconstruction through NeuralRecon, and human intervention, all implemented on ROS2. 
        The right block shows the virtual 3D scene built on Unity.}
        \label{fig:block_diagram}
    \end{center}
\end{figure*}

Simply combining the coverage input $u_{ci}$ with the human command $u_{h}$ as in (\ref{eq:actual_control_input}) can lead to
conflicts. Specifically, coverage control input $u_{ci}$ affects the average pose of
the drones to be operated by the human, and hence could disturb the manual
operations of the drones, which may degrade the system stability and the human operability.
To address this issue, we present a stealthy control strategy that decouples the
two inputs based on the notion of the redundancy-based control framework using the nullspace.
% The objective of stealthy control is that the drone motions driven
% by the coverage control are invisible from the human operator.
%%============================================================================%%
% \subsection{Definition of Stealthy Control}
%%============================================================================%%

To formulate the stealthy control, we begin by describing the collective system of $n$ drones in
a compact vectorized form. Define the concatenated state vector as $G\in \mathbb{R}^{5n}$, 
the human input vector as $U_{h}\in \mathbb{R}^{5n}$, and
the autonomous input vector as $U_{a}\in \mathbb{R}^{5n}$ as:
\begin{align*}
    G = \begin{bmatrix}p\\[3pt] \theta\end{bmatrix},\quad 
    U_{h}= \begin{bmatrix}U_{h}^{p}\\[3pt] U_{h}^{\theta}\end{bmatrix},\quad
    U_{a}= \begin{bmatrix}U_{a}^{p}\\[3pt] U_{a}^{\theta}\end{bmatrix},
\end{align*}
where
\begin{align*}
    % P_{\text{Pos}}  & = \begin{bmatrix}p_{1}^{\top}&p_{2}^{\top}&\cdots&p_{n}^{\top}\end{bmatrix}^{\top},
    % P_{\text{Ang}}  = \begin{bmatrix}\theta_{1}^{\top}&\theta_{2}^{\top}&\cdots&\theta_{n}^{\top}\end{bmatrix}^{\top},    \\ %
    p      & = \begin{bmatrix}{p_1}^{\top}&{p_2}^{\top}&\cdots&{p_n}^{\top}\end{bmatrix}^{\top}\in\mathbb{R}^{3n},                \\ %
    \theta      & = \begin{bmatrix}{\theta_1}^{\top}&{\theta_2}^{\top}&\cdots&{\theta_n}^{\top}\end{bmatrix}^{\top}\in\mathbb{R}^{2n},                \\ %
    U_{h}^{p}      & = \begin{bmatrix}{u_{h}^p}^{\top}&{u_{h}^p}^{\top}&\cdots&{u_{h}^p}^{\top}\end{bmatrix}^{\top}\in\mathbb{R}^{3n},                   \\ %
    U_{h}^{\theta} & = \begin{bmatrix}{u_{h}^\theta}^{\top}&{u_{h}^\theta}^{\top}&\cdots&{u_{h}^\theta}^{\top}\end{bmatrix}^{\top}\in\mathbb{R}^{2n}, \\ %
    U_{a}^{p}      & = \begin{bmatrix}{u_{a1}^p}^{\top}&{u_{a2}^p}^{\top}&\cdots&{u_{an}^p}^{\top}\end{bmatrix}^{\top}\in\mathbb{R}^{3n},                \\ %
    U_{a}^{\theta} & = \begin{bmatrix}{u_{a1}^\theta}^{\top}&{u_{a2}^\theta}^{\top}&\cdots&{u_{an}^\theta}^{\top}\end{bmatrix}^{\top}\in\mathbb{R}^{2n}.   
    % U_{a}^p = \begin{bmatrix}{u_{a1}^p}\\{u_{a2}^p}\\\vdots\\{u_{an}^p}\end{bmatrix},%
    % U_{a}^\theta = \begin{bmatrix}{u_{a1}^\theta}\\{u_{a1}^\theta}\\\vdots\\{u_{a1}^\theta}\end{bmatrix},%
    % U_{h}^p = \begin{bmatrix}{u_{h}^p}\\{u_{h}^p}\\\vdots\\{u_{h}^p}\end{bmatrix},%
    % U_{h}^\theta = \begin{bmatrix}{u_{h}^\theta}\\{u_{h}^\theta}\\\vdots\\{u_{h}^\theta}\end{bmatrix}.
\end{align*}
% \begin{align*}
%     P = \begin{bmatrix}P_{\text{Pos}}\\ P_{\text{Ang}}\end{bmatrix},\  U_{a}= \begin{bmatrix}U_{a\text{Pos}}\\ U_{a\text{Ang}}\end{bmatrix},\  U_{h}= \begin{bmatrix}U_{h\text{Pos}}\\ U_{h\text{Ang}}\end{bmatrix},\
% \end{align*}
% where
% \begin{align*}
%     P_{\text{Pos}}  & = \begin{bmatrix}\text{Pos}(p_{1})^{\top}&\text{Pos}(p_{2})^{\top}&\cdots&\text{Pos}(p_{n})^{\top}\end{bmatrix}^{\top},    \\ %
%     P_{\text{Ang}}  & = \begin{bmatrix}\text{Ang}(p_{1})^{\top}&\text{Ang}(p_{2})^{\top}&\cdots&\text{Ang}(p_{n})^{\top}\end{bmatrix}^{\top},    \\ %
%     U_{a\text{Pos}} & = \begin{bmatrix}\text{Pos}(u_{a1})^{\top}&\text{Pos}(u_{a2})^{\top}&\cdots&\text{Pos}(u_{an})^{\top}\end{bmatrix}^{\top}, \\ %
%     U_{a\text{Ang}} & = \begin{bmatrix}\text{Ang}(u_{a1})^{\top}&\text{Ang}(u_{a2})^{\top}&\cdots&\text{Ang}(u_{an})^{\top}\end{bmatrix}^{\top}, \\ %
%     U_{h\text{Pos}} & = \begin{bmatrix}\text{Pos}(u_{h})^{\top}&\text{Pos}(u_{h})^{\top}&\cdots&\text{Pos}(u_{h})^{\top}\end{bmatrix}^{\top},    \\ %
%     U_{h\text{Ang}} & = \begin{bmatrix}\text{Ang}(u_{h})^{\top}&\text{Ang}(u_{h})^{\top}&\cdots&\text{Ang}(u_{h})^{\top}\end{bmatrix}^{\top}.
% \end{align*}
% To simplify the subsequent equations, these vectors are partitioned into
% positional and angular components.
The collective dynamics of the full $5n$-dimensional system can then be written in
the following form:
\begin{align}
    \label{eq:stealth_control_dynamics} %
    \dot G=U_{h}+U_{a}.
\end{align}
Let us now define the stealthy control as follows.
\begin{definition}
    Consider a dynamical system %$G_{P_0}$ 
    having the state equation (\ref{eq:stealth_control_dynamics}),
    the control output $\bar{g}$, the control input $U_{a}$, and the external
    input $U_{h}$. Then, the input signal $U_{a}(\cdot)$ is said to be a
    stealthy control if
    $\bar g(t; U_{a}(\cdot), U_{h}(\cdot), G_{0}) = \bar g(t; 0, U_{h}(\cdot), G_{0}
    )$
    holds for all time $t$, all initial states $G_{0}$, and external input signal
    $U_{h}(\cdot)$, where $\bar g(t; U_{a}(\cdot), U_{h}(\cdot), G_{0})$ denotes
    the output at time $t$ from the initial state $G_{0}$ when the input signals
    $U_{a}(\cdot), U_{h}(\cdot)$ are applied to the system.
    \label{def:stealth_control}
\end{definition}
Suppose that the stealthy control is applied to the system. Then, %the actual drone $i$ follows the dynamics (\ref{eq:actual_control_input}),
%whereas 
the average state $\bar{g}$ perceived by the operator is governed by the
simpler dynamics $\dot{\bar{g}}= u_{h}$.
In other words, the motion of the drones governed by coverage control is imperceptible to the human operator.
%%============================================================================%%
% \subsection{Stealthy Coverage Control Design}
%%============================================================================%%
% Note that the averaging processes for
% drone camera positions and orientations differ, as described in (\ref{eq:definition-Pos})
% and (\ref{eq:definition-Dir}). %
% This inherent difference requires the $5n$D system to be divided into
% translational and rotational motions, respectively.

In order to design the stealthy control, we multiply a
matrix $A(g)\in\mathbb{R}^{5n\times5n}$ to the collection of coverage control
input $u_{ci}$ for all drones, denoted by $U_{c}$, as below:
% so that the autonomous control input $U_{a}$ satisfies below:
\begin{align}
    \label{eq:stealth_control_input}                                                                                                  %
    % U_{a} = A(\bar{g})U_{c}:= \begin{bmatrix}A_{\text{Pos}}&O \\ O&A_{\text{Ang}}(\bar{g})\end{bmatrix} \begin{bmatrix}U_{c}^p\\ U_{c}^\theta\end{bmatrix},
    U_{a}= A(g)U_{c}:= \begin{bmatrix}A^{p}&O \\ O&A^{\theta}(g)\end{bmatrix} \begin{bmatrix}U_{c}^{p}\\[3pt] U_{c}^{\theta}\end{bmatrix},
\end{align}
where
\begin{align*}
    % U_{c\text{Pos}} & = \begin{bmatrix}\text{Pos}(u_{c1})^{\top}&\text{Pos}(u_{c2})^{\top}&\cdots&\text{Pos}(u_{cn})^{\top}\end{bmatrix}^{\top}, \\ %
    % U_{c\text{Ang}} & = \begin{bmatrix}\text{Ang}(u_{c1})^{\top}&\text{Ang}(u_{c2})^{\top}&\cdots&\text{Ang}(u_{cn})^{\top}\end{bmatrix}^{\top}.
    U_{c}^{p}      & = \begin{bmatrix}{u_{c1}^p}^{\top}&{u_{c2}^p}^{\top}&\cdots&{u_{cn}^p}^{\top}\end{bmatrix}^{\top}\in\mathbb{R}^{3n},                \\ %
    U_{c}^{\theta} & = \begin{bmatrix}{u_{c1}^\theta}^{\top}&{u_{c2}^\theta}^{\top}&\cdots&{u_{cn}^\theta}^{\top}\end{bmatrix}^{\top}\in\mathbb{R}^{2n}. %
\end{align*}
% where $U_c\in\mathbb{R}^{5n}$ is the coverage control input vector.
\begin{theorem}
    Consider the system in Definition~\ref{def:stealth_control}. %
    Denote $i$-th column of $\bm1_{n}\otimes I_{3} \in \mathbb{R}^{3n\times 3}$ and $\left(\dfrac{\partial}{\partial \theta}\overline{\mathrm{Dir}(g)}\right)^\top$ by $e_i^p$ and $e_i^\theta({g})\in \mathbb{R}^{2n\times 3}$, respectively.
    Suppose that the matrix $A(g)$ in (\ref{eq:stealth_control_input}) is designed
    so that
    \begin{align}
        \normalfont e_i^{p}\in\text{ker}(A_p^\top),   \quad e_i^{\theta}(g)\in\text{ker}\Big((A^\theta({g}))^\top\Big)\ \forall i =1,2,3.
\label{eqn:hatanaka_edit}
    \end{align}
    % to satisfy (\ref{eq:stealth_control_dynamics}). %
    Then, the input (\ref{eq:stealth_control_input}) constitutes a
    stealthy control.
\end{theorem}
% \begin{proof}
{\it Proof.}
    We consider the time derivative of the average position as follows.
    \begin{align*}
        % \dot{\bar{p}} & = \dfrac{1}{n}\sum_{i\in\mathcal{I}} \dot{p}_i = \dfrac{1}{n}\sum_{i\in\mathcal{I}}\dot{P}_{\mathrm{Pos}} = \dfrac{1}{n}\bm1_{n}^{\top}\otimes I_{3}U_{h\text{Pos}}.
        \dot{\bar{p}} & = \dfrac{1}{n}\sum_{i\in\mathcal{I}}\dot{p}_{i}= \dfrac{1}{n}\bm1_{n}^{\top}\otimes I_{3}\dot{p}                          \\
                      & = \dfrac{1}{n}\bm1_{n}^{\top}\otimes I_{3}(U_{h}^{p}+A^{p}U_{c}^{p}) = \dfrac{1}{n}\bm1_{n}^{\top}\otimes I_{3}U_{h}^{p}.
    \end{align*}
    Meanwhile, the time derivative of the average orientation is as follows.
    \begin{align*}
        % \dot{\bar{r}} & = \dfrac{\partial}{\partial P_{\text{Ang}}}\bar{r}\dot{P}_
        % {\text{Ang}} = \dfrac{\partial}{\partial P_{\text{Ang}}}\bar{r}U_{h\text{Ang}}.
        \dot{\overline{\mathrm{Dir}(g)}} &= \dfrac{\partial}{\partial \theta}\overline{\mathrm{Dir}(g)}\dot\theta
        = \dfrac{\partial}{\partial \theta}\overline{\mathrm{Dir}(g)}(U_{h}^{\theta}+A^{\theta}({g})U_{c}^{\theta})\\
        &=\dfrac{\partial}{\partial \theta}\overline{\mathrm{Dir}(g)}U_{h}^{\theta}.
    \end{align*}
    % \begin{align*}
    %     \dot{\overline{\mathrm{Pos}(p)}} & = \dfrac{1}{n}\sum_{i=1}^{n} \dot{\mathrm{Pos}(p_i)} = \dfrac{1}{n}\sum_{i=1}^{n}\dot{P}_{\mathrm{Pos}} = \dfrac{1}{n}\bm1_{n}^{\top}\otimes I_{3}U_{h\text{Pos}}.
    % \end{align*}
    % Meanwhile, the time derivative of the average orientation is as follows.
    % \begin{align*}
    %     \dot{\overline{\mathrm{Dir}(p)}} & = \dfrac{\partial}{\partial P_{\text{Ang}}}\overline{\mathrm{Dir}(p)}\dot{P}_{\text{Ang}} = \dfrac{\partial}{\partial P_{\text{Ang}}}\overline{\mathrm{Dir}(p)}U_{h\text{Ang}}.
    % \end{align*}
    Therefore, $U_{a}$ does not contribute to the evolution of the control output
    $\bar{g}$, and $U_{a}$ is a stealthy control.
% \end{proof}
\hfill $\square$

In designing a stealthy control, we have to fix a matrix $A(g)$ meeting
(\ref{eqn:hatanaka_edit}).
Now, denote $V = \bm1_{n}\otimes I_{3} \in \mathbb{R}^{3n\times 3}$. We also define $W(g)\in \mathbb{R}^{2n\times 2}$ by eliminating an arbitrary column from the matrix $\left(\dfrac{\partial}{\partial \theta}\overline{\mathrm{Dir}(g)}\right)^\top\in \mathbb{R}^{2n\times 3}$.
Then, the matrices that satisfy (\ref{eqn:hatanaka_edit}) are given, for example, as 
\begin{align}
A^p &= I_{3n} - V(V^\top V)^{-1}V^\top,\\
A^\theta(g) &= I_{2n} - W(g)(W^\top(g) W(g))^{-1}W^\top(g).
\end{align}
Hereafter, we shall mean these specific matrices
when using the symbol $A(g)$.

We are now ready to present the stealthy coverage control.
We first reformulate the collection of the constraints in (\ref{eq:coverage_control}) for all drones into the following form:
\begin{align}
    \label{eq:coverage_control_mat} %
    % U_{c}^{*}     & = \underset{U_c\in\mathbb{R}^{5n}}{\text{arg min}\ \|U_c\|^2} \\
    % \text{s.t.}\  & GU_{c}\leq F,
    BU_{c}\geq C,
\end{align}
where
\begin{align*}
    B     & =\dfrac{\partial b_I}{\partial g},\quad b_{I}=\begin{bmatrix}b_{1,I}\  b_{2,I}  \cdots\ b_{n,I}\end{bmatrix}^{\top}, \\
    C     & =\begin{bmatrix}c_{1}\ c_{2}\ \cdots\ c_{n}\end{bmatrix}^{\top},                                                    \\
    c_{i} & =-\sum_{j\in\mathcal{V}_i(g)}\bigg(\dfrac{\partial b_{i,I}}{\partial{\phi_j}}\bigg)\dot\phi_{j}-\alpha(b_{i,I}).
\end{align*}
% To introduce the stealthy control into this centralized formulation,
% we modify (\ref{eq:coverage_control_mat}) as follows:
To ensure that the quadratic program provides the stealthy control, 
we multiply the matrix $A({g})$ to $U_c$ in (\ref{eq:coverage_control_mat}). 
We then modify the quadratic program as follows:
% However, it is not appropriate to directly use the solution $U_{c}^{*}$ of the quadratic
% programming problem as the autonomous input of the system, since the actual
% control input $A(p)U_{c}^{*}$ applied to the drones cannot be guaranteed to
% satisfy the problem constraints. %
% Accordingly, we first reformulate the quadratic programming problem (\ref{eq:coverage_control_mat})
% as follows.
% \begin{subequations}
    \begin{align}
        U_{c}^{*}     & = \underset{U_c}{\text{arg min}}\ \|U_c\|^2 \notag\\
        \text{s.t.}\  & BA({g})U_{c}\geq C.
        \label{eq:stealthy_coverage_control} %
    \end{align}
% \end{subequations}
Then, the autonomous control input $U_a = A({g})
U_{c}^{*}$ always constitutes a stealthy control, while ensuring the performance constraint.
% Fig.~\ref{fig:block_diagram} illustrates the control architecture of the
% proposed stealthy coverage control. 
The overall system architecture is illustrated in Fig.~\ref{fig:block_diagram}.
Note that the stealthy control (\ref{eq:stealth_control_dynamics})
is inherently centralized, 
and distributed implementation of the present controller
is left as future work.

Besides avoiding the interference of coverage control with manual control, the
present control architecture has another advantage. Namely, concerns about
stability of the human-in-the-loop system can be decoupled from designing
efficient image sampling strategy. The remaining stability-related issue will be
presented in a separate paper, but the passivity-based paradigm presented 
by \cite{Atman:2019} and \cite{Hatanaka:2024}
allows one to prove stability under the assumption of human passivity or passivity shortage.

%
%%============================================================================%%
\section{Simulation}
\label{chp:simulation}
%%============================================================================%%
In this section, we run two simulations to demonstrate the effectiveness of the present semi-autonomous control architecture. 
%proposed stealthy coverage control successfully resolves the potential conflicts between coverage control and human operation. 
% We first validate this by showing that the drone's average
% state responds to the human command even while the coverage control is being
% applied. 
% Furthermore, we evaluate the overall system performance under human
% guidance with visual mesh feedback in real time.
% We compare the quality of the resulting 3D meshes obtained from the 
% coverage control and proposed stealthy coverage control.
As shown in Fig.~\ref{fig:block_diagram}, 
the simulation studies were conducted on a simulator built based on Robot Operating System 2 (ROS2) and Unity.
The control algorithm including the 3D structure reconstruction through NeuralRecon was implemented on ROS2.
Unity replicated real indoor environment with a wardrobe, table and two sofas, and simulated the image sampling therein.

%s are implemented using Robot Operating System (ROS) 2, 
%and the reconstructed 3D mesh is visualized in a Unity
%indoor model including a wardrobe, table, and two sofas.
% The simulations are performed using Robot Operating System (ROS)2, and
% the simulation architechture is shown in Fig.~\ref{fig:block_diagram}.
% For 3D mesh visualization, a simulation environment is constructed in Unity,
% as shown in Fig.~\ref{fig:block_diagram},
% featuring an indoor environment with sofa, wardrobe and table.
% as shown in Fig.~\ref{fig:unity}
We employed Meta Quest 3S (Meta Platforms, Inc.) as an interface between the human operator and the drones.
The reconstructed 3D model and the average pose of the drones displayed in Unity are projected onto the HMD. A human operator determines the velocity commands $u_h$ via the VR controller based on these information. 
Note that we added an additional function such that the operator can keep sending the commands when pressing a button in order to avoid unintentional commands. 
The position at the moment of pressing the button defines a temporary origin of the interface frame, and the displacement from this origin is mapped to translational velocity command $u_h^p$.
The angular velocity $u_h^\theta$ is commanded by deflecting the joystick,
where the vertical and horizontal axes are mapped to the pitch 
and yaw angular velocity commands, respectively.
The commands from the VR controller are scaled to the range of 
$[-0.5,0.5]$ m/s for the translational velocity and 
$[-0.15,0.15]$ rad/s for the angular velocity.
The commands are then sent to ROS2, and added to the control input as (\ref{eq:actual_control_input}).
%The translational velocity command  is specified by the position of the handheld controller relative to the origin of the interface frame. The angular velocity $u_h^\theta$ is commanded by the orientation of the controller relative to the interface frame.

%Releasing the button and joystick brings each commanded velocity to zero.

In the present simulations, the number of drones was set to $n=3$.
The extended class-$\mathcal{K}$ function $\alpha(\cdot)$ was defined as a linear function 
$\alpha(b_{i}) = a b_{i,I}\ (a = 1.0)$.
The target reconstruction field $\mathcal{Q}$ was set as
$[-4, 4]\text{m}\times[-4, 4]\text{m}\times[0.4, 1.2]\text{m}$,
and the drone flight field $\mathcal{P}$ was defined as 
$[-4, 4]\text{m}\times[-4, 4]\text{m}\times[2.0, 2.4]\text{m}$.
% , with $\Theta^{h}=[-\pi, \pi]$ and $\Theta^{v}= (0, \pi/2]$.
The target field $\mathcal{Q}$ was discretized into $m=7\times10^{5}$ cells, 
each with a volume of
$0.2\text{m}\times 0.2\text{m}\times0.2\text{m}\times0.3\text{rad}\times0.3\text{rad}$.
In addition to the conditions in (\ref{eq:stealthy_coverage_control}),
constraints on the flight field $\mathcal{P}$ and  
the gimbal pitch angle $\varphi_i^v$ were imposed as described in \cite{Hanif:2025},
and the coverage control inputs were limited to $|u_{ci}|\leq0.05$.
% we imposed additional constraints to ensure that 
% the drones remain within the flight field $\mathcal{P}$ and 
% to enforece limits on the gimbal pitch angle $\varphi_i^v$, as described in \cite{Hanif:2025}.
The parameters were set as follows: $D=1.4$, $\sigma_{1}=0.15$, $\sigma_{2}=0.15$,
$\sigma_{3}=0.3$, $\gamma=0.0004$, $\delta=0.5$, $\overline\sigma_{1}=0.
12$, $\overline\sigma_{2}=0.12$.
%%============================================================================%%
\subsection{Demonstration of Stealthy Control}
%%============================================================================%%
\begin{figure}[t]
  \begin{center}
    \includegraphics[width=0.75\columnwidth]{
      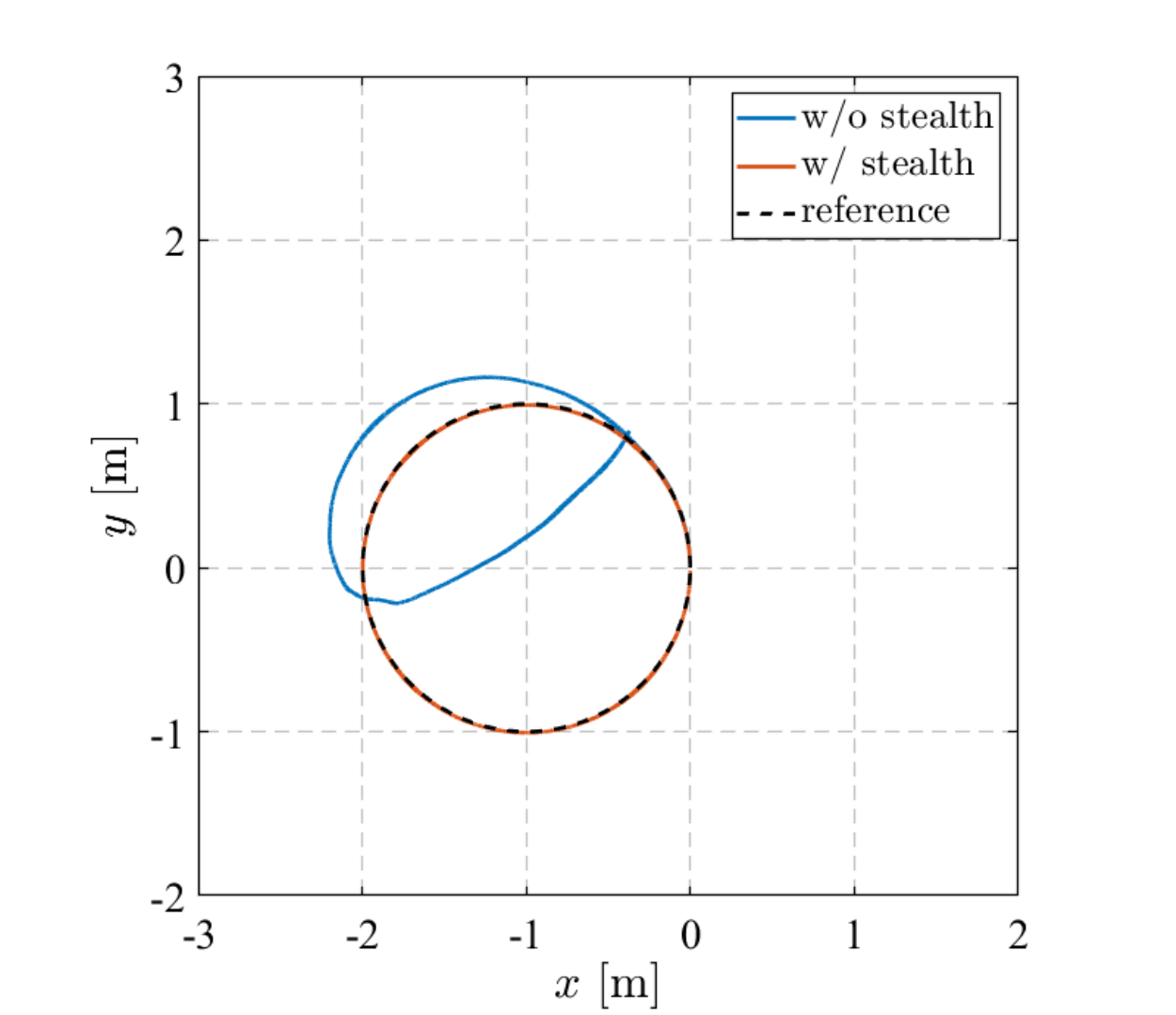
    }
    \caption{Comparison of the trajectories of the average position for a circular path with and without
    the stealthy control.}
    \label{fig:circular_path}
  \end{center}
\end{figure}

The first simulation was conducted in order to demonstrate the effectiveness of the stealthy coverage control.
To this end, we applied the following $u_{h}$ generating a circular trajectory without using manual control.
\begin{align*}
  u_{h}=\begin{bmatrix}-\omega\sin\omega t & \omega\cos\omega t & 0 & 0 & 0\end{bmatrix}^{\top},
\end{align*}
where $\omega$ is the angular velocity of the circular trajectory.
In this simulation, we took $\omega = 0.05$rad/s.

Let us confirm the evolution of the average state $\bar g$ while activating the coverage control.
The simulation results are shown
in Figs.~\ref{fig:circular_path} and \ref{fig:circular_trajectory_ori},
where the initial average state was set to $\bar{g}=[0\ 0\ 2.2\ 0\ 0\ 1]^\top$.
Fig.~\ref{fig:circular_path} %and \ref{fig:circular_trajectory_pos} show the
shows the trajectories of the average position with (red) and without the stealthy control (blue).
%, respectively.
Without the stealthy control, the average position is drifted by coverage control and it fails to track the circular path, which would cause the difficulties in the human operation.
%deviating from it due to the conflicts between the
%coverage control and command $u_h$.
In contrast, the proposed stealthy coverage control enables the average position
to exactly follow the desired path.
% Fig. \ref{fig:circular_trajectory_ori} shows the time response of each element of the optical axis of the average drone $\overline{\mathrm{Dir}(g)}$ with (bottom) and without the stealthy control (top).
Fig. \ref{fig:circular_trajectory_ori} shows the time response of the deviations of the average optical axis from the initial values.
It is confirmed from this figure that 
the average axis remains constant with the stealthy coverage control, while they are slightly drifted without the stealthy control.
These results demonstrate that the stealthy coverage control successfully
decouples the manual control and the autonomous coverage control.

%%============================================================================%%
\subsection{Demonstration of 3D  Reconstruction}
%%============================================================================%%

In this subsection, we conducted a human-in-the-loop simulation in order to demonstrate the practical benefit of the human intervention for the 3D structure reconstruction. A video of the present simulation can be found at
\url{https://youtu.be/I8UZhCcUvb4}.
\begin{figure}[t]
  \centering
          \includegraphics[width=1\columnwidth]{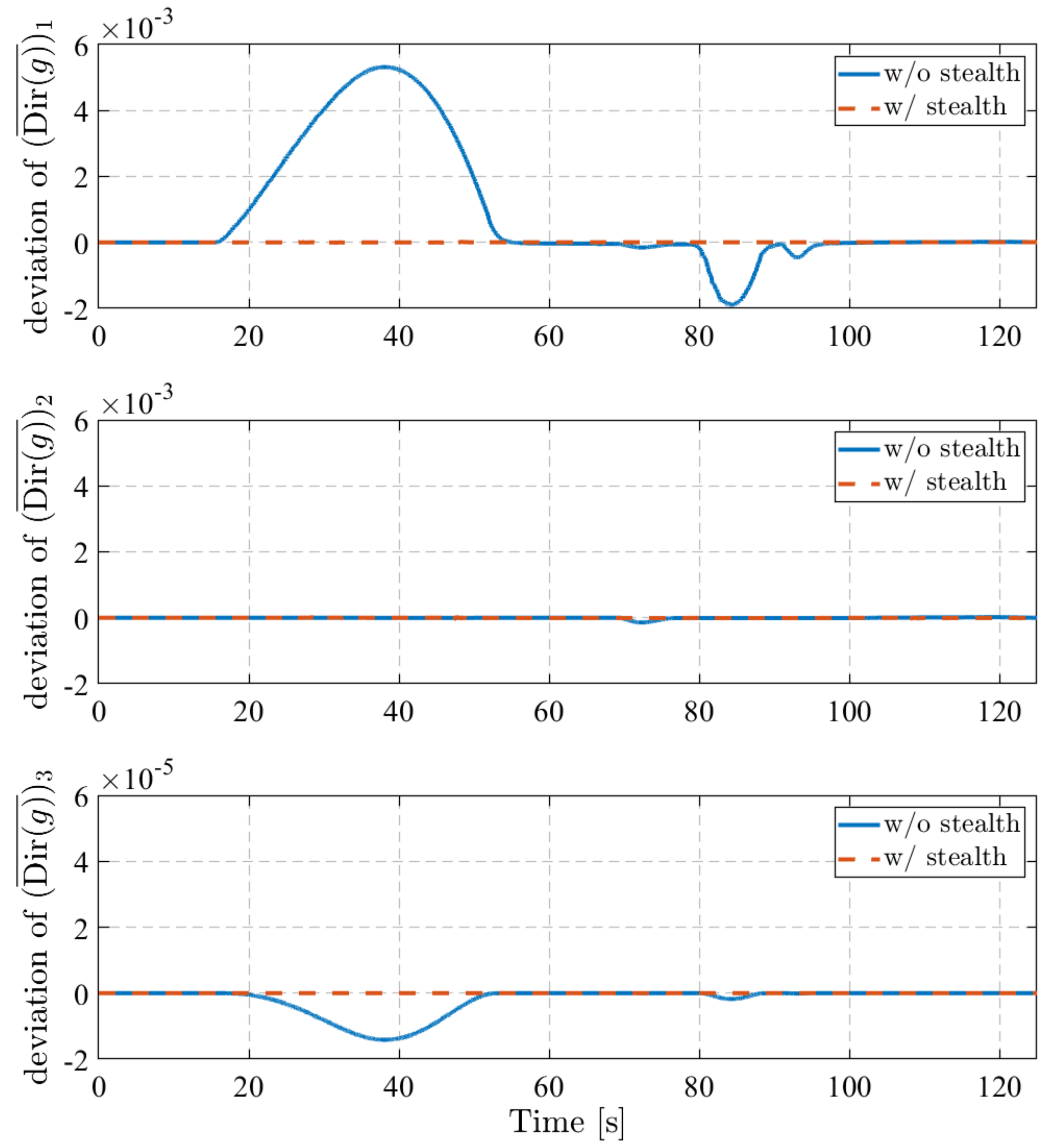}
          \caption{Time responses of the deviations of the average optical axis from their initial values.}
          \label{fig:circular_trajectory_ori}
  \end{figure}

Fig.~\ref{fig:stealth_snap_mesh} presents snapshots of the simulation.
The colormaps in the top figures represent the values of $\phi_{j}$ (red: high, blue: low), and the
white/blue arrows indicate the average state of the drones.
In the beginning, the human operator does not provide any command, 
allowing the system to progressively reconstruct the environment up to $t=600$s through autonomous coverage control and NeuralRecon.
We observe that the average pose is invariant during the period due to the stealthy control.
After $t=600$s, the operator identifies areas with insufficient reconstruction accuracy such as the side and upper surfaces of the furniture items, and 
guides the average pose of the drones toward such regions until $t=2100$s.
We see that the quality of the model around the top surfaces of the sofas and the wardrobe is enhanced through intensive image sampling around these specific areas by manual control.
% As coverage progresses over time, areas of higher importance appear around the average
% values ($t=300$s).
% Consequently, the drones begin to focus their observations
% near the average state ($t=900$s, $t=1200$s). 
% After further time, this focused observation expands, 
% leading to coverage of the entire environment ($t=2100$s).
% Thus, the stealthy coverage control successfully performs its coverage task,
% achieving both overall coverage and focused observations near the average values.
\begin{figure*}[t]
  \begin{center}
      \includegraphics[width=15.3cm]{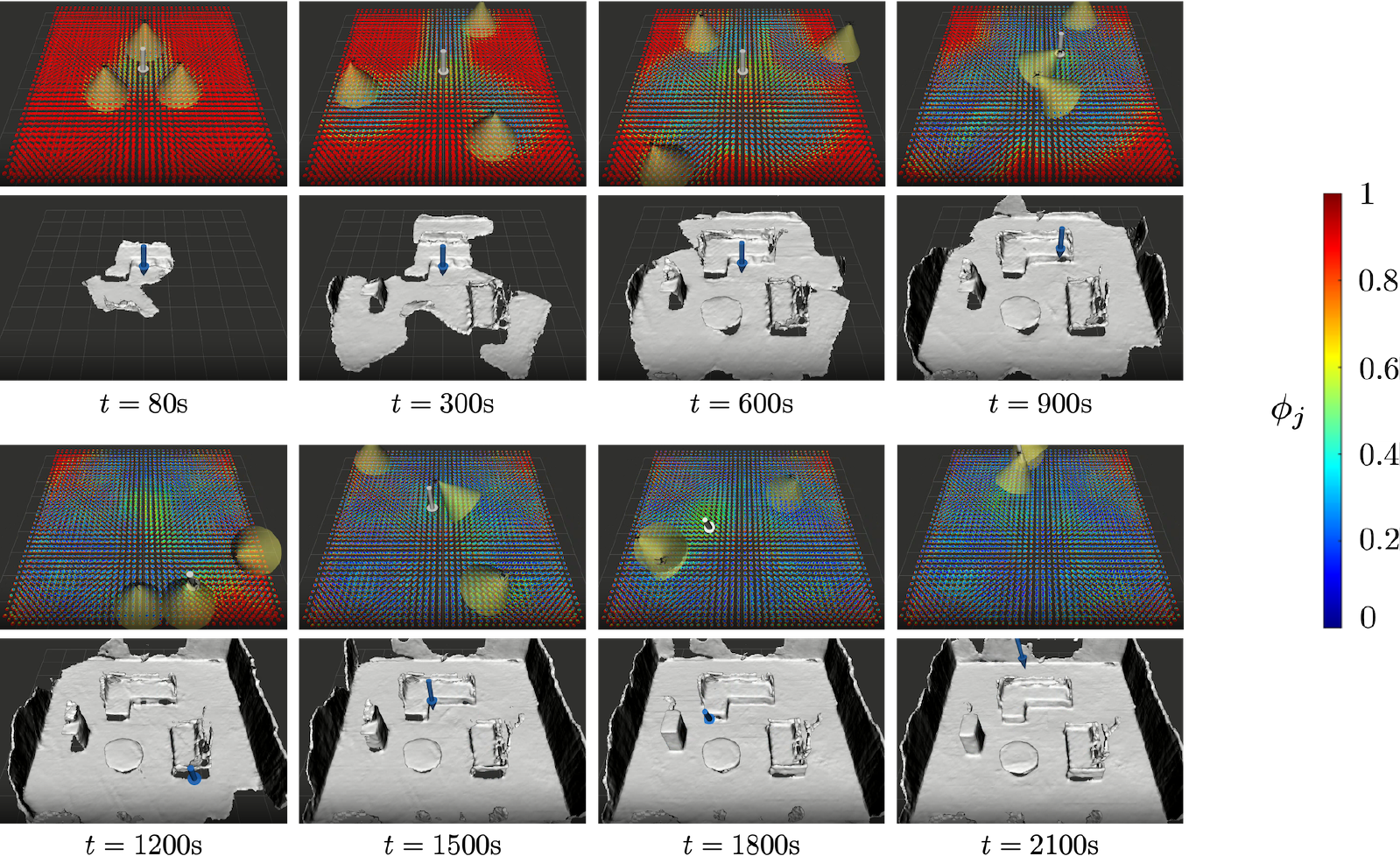}
      \caption{Snapshots of the three drones simulation with stealthy coverage control. 
        The snapshots show the evolution of the reconstructed structure and 
        the distribution of $\phi_{j}$ over time.}
      \label{fig:stealth_snap_mesh}
  \end{center}
\end{figure*}
\begin{figure*}[tbp]
  % \centering
  %--- (a) ---
  \begin{minipage}{0.3\linewidth}
    \centering
    \includegraphics[width=\linewidth]{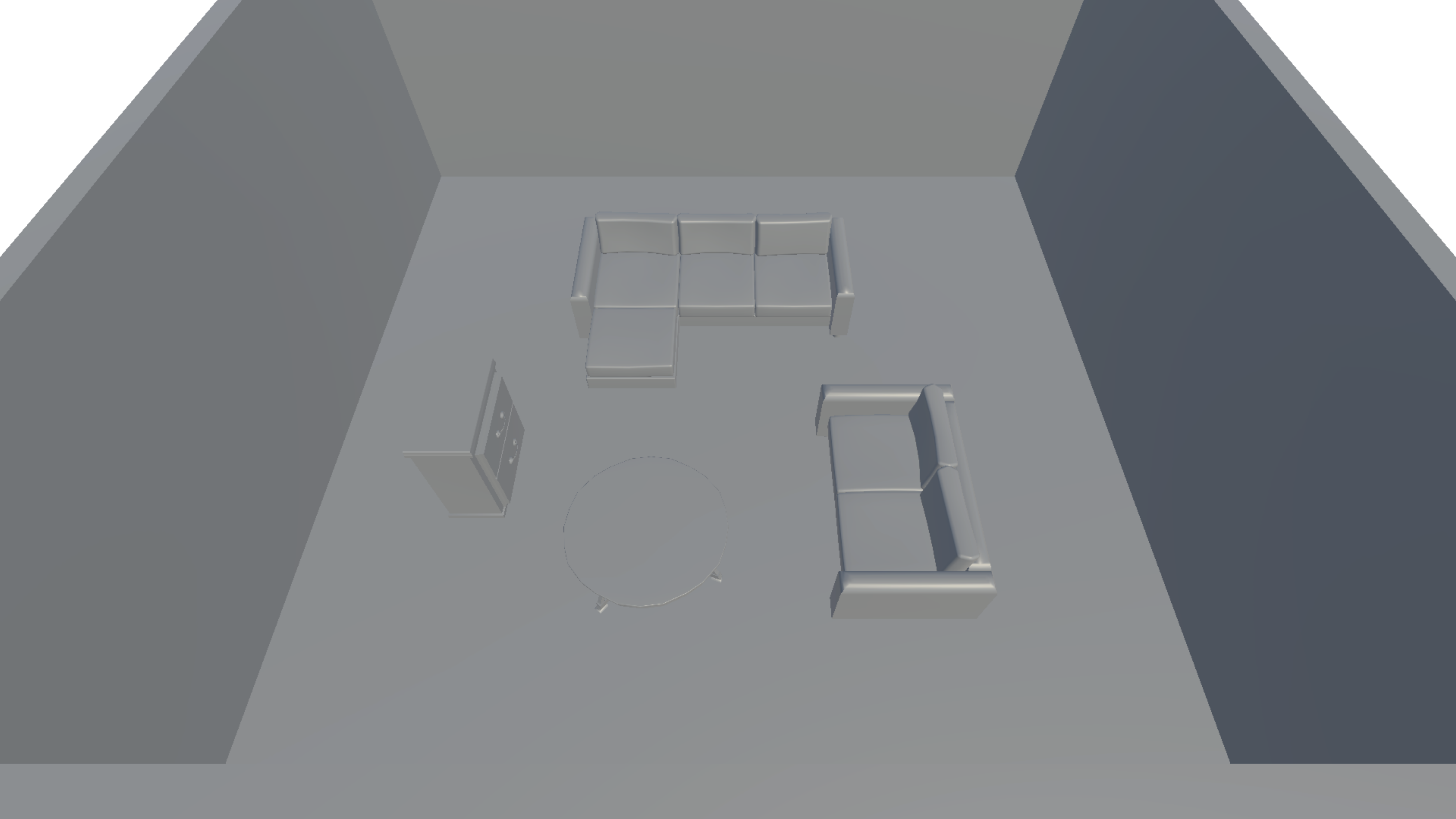}
    \subcaption{} \label{fig:ground_truth}
  \end{minipage}
  \hfill
  %--- (b) ---
  \begin{minipage}{0.3\linewidth}
    \centering
    \includegraphics[width=\columnwidth]{
      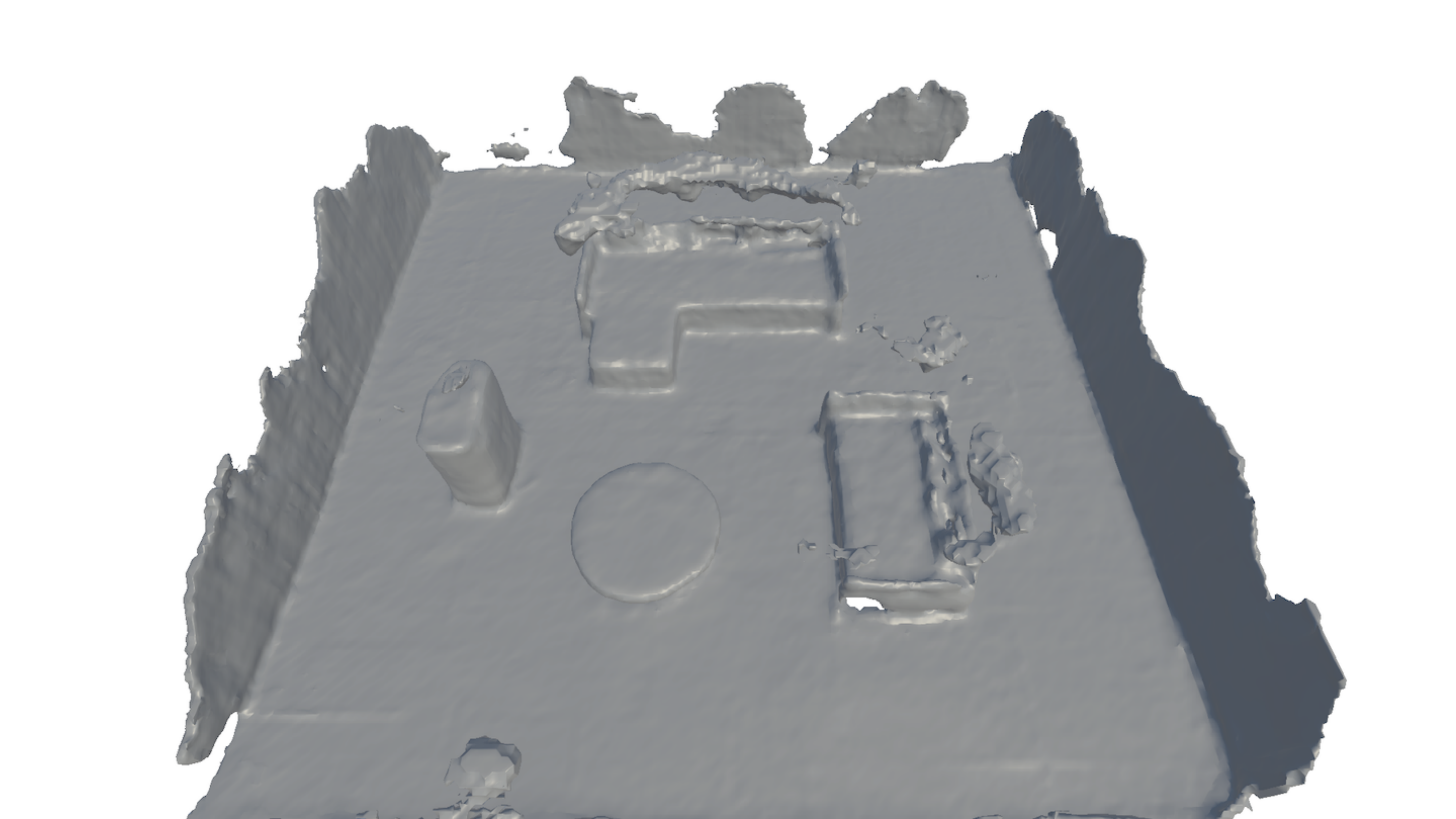
    }
    \subcaption{} \label{fig:mesh_comparison_without}
  \end{minipage}
  \hfill
  %--- (c) ---
  \begin{minipage}{0.3\linewidth}
    \centering
    \includegraphics[width=\columnwidth]{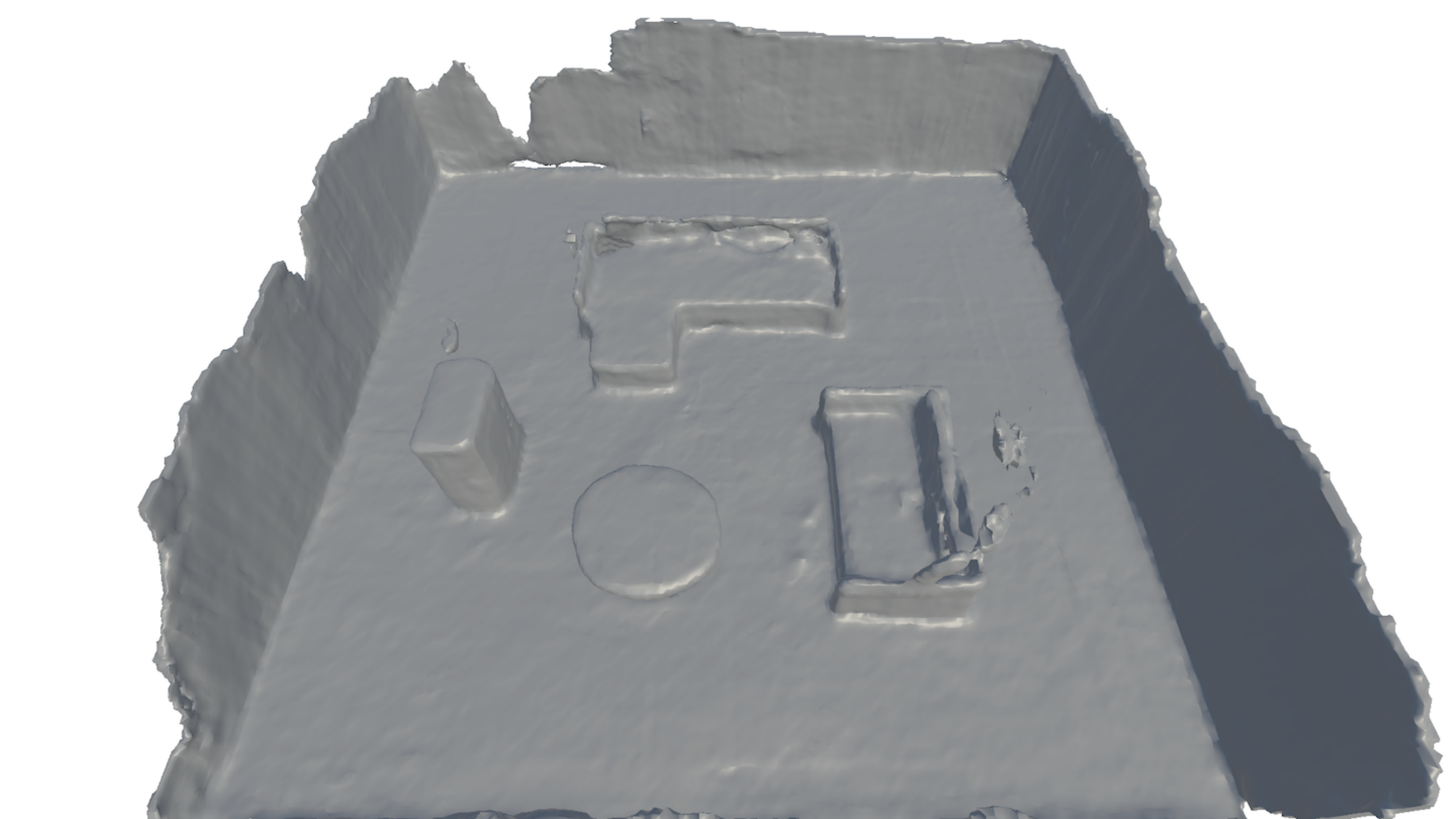}
    \subcaption{} \label{fig:mesh_comparison_with}
  \end{minipage}
  \caption{Comparison between reconstructed models: (a) Ground truth model in Unity, (b) reconstructed model without
  the human intervention, (c) reconstructed model with the human intervention.}
  \label{fig:mesh_comparison}
\end{figure*}

Fig.~\ref{fig:mesh_comparison_without} and Fig.~\ref{fig:mesh_comparison_with}
compare the final 3D models generated from the fully autonomous coverage control and
the proposed semi-autonomous system with stealthy coverage control, respectively.
Note that both simulations ended when $J$ was close to zero and a human determined that the model would not change any further.
As compared to the ground truth model in Fig. \ref{fig:ground_truth}, the model in Fig.~\ref{fig:mesh_comparison_without} contains some objects in the air above the top of the sofa that do not exist in practice. In addition, the mesh on the top surface of the wardrobe is also slightly distorted.
On the other hand, the model generated by the proposed method is visibly more detailed,
in particular, non-existent objects have been almost removed, and the top surface of the wardrobe gets nearly rectangular.
This improvement in accuracy is attributed to the human interventions and the resulting intensive image sampling around the areas having high structural complexity.

In summary, we conclude that the present semi-autonomous control architecture achieves complexity-aware flexible image sampling, 
which contributes to the quality enhancement of the reconstructed 3D model.
\section{Conclusion}
\label{chp:conclusion}
%%============================================================================%%
In this paper, we proposed a semi-autonomous image sampling strategy for 3D reconstruction with adaptation to structural complexity of the scene by leveraging the flexible reasoning and situational recognition capabilities of humans.
In the present system, a human operator identifies areas that needs more images for improving the model quality based on the real-time 3D model, and then navigates the average pose of the drones to such areas.
To this end, we developed a novel coverage control reflecting the human intention on the location to be sampled more. 
Subsequently, in order to avoid operational conflicts between manual control and autonomous coverage control, 
we designed a novel stealthy coverage control that decoupled the drone motion for efficient image sampling from human navigation.
Simulation studies on a Unity/ROS2-based simulator demonstrated that the present semi-autonomous system outperformed the one without the human intervention in the sense of the reconstructed model quality.

There are several issues that remain to be addressed in the future.
First, we need to analyze closed-loop stability of the human-in-the-loop system. 
Additionally, human modeling and analysis must be conducted to examine whether the operator meets stability conditions in theory.
%Associated with this issue, whether human passivity/passivity shortage revealed in \cite{Atman:2019} and \cite{Hatanaka:2024} is applied to the present case including rotational control should be discussed based on the human operation data. 
Distributed implementation of the stealthy coverage control is also left as future work.
We also have to conduct experimental studies beyond the idealized simulation in order to reveal the practical benefit of the present control strategy.

Approval of all ethical and experimental procedures and protocols was granted by the Human Subjects Research Ethics Review Committee in Institute of Science Tokyo under Application No. 2025151, and performed in line with the Helsinki on Ethical
Principles for Medical Research and Ethical Guidelines for Medical and Biological Research Involving Human Subjects. %

% \begin{ack}
% Place acknowledgments here.
% \end{ack}

\bibliography{ifacconf} % bib file to produce the bibliography

\begin{thebibliography}{24}
\providecommand{\natexlab}[1]{#1}
\providecommand{\url}[1]{\texttt{#1}}
\providecommand{\urlprefix}{URL }
\expandafter\ifx\csname urlstyle\endcsname\relax
  \providecommand{\doi}[1]{doi:\discretionary{}{}{}#1}\else
  \providecommand{\doi}{doi:\discretionary{}{}{}\begingroup \urlstyle{rm}\Url}\fi

\bibitem[{Ames et~al.(2017)Ames, Coogan, Egerstedt, Notomista, Sreenath, and Tabuada}]{Ames:2017}
Ames, A.D., Coogan, S., Egerstedt, M., Notomista, G., Sreenath, K., and Tabuada, P. (2017).
\newblock Control barrier function based quadratic programs for safety critical systems.
\newblock \emph{IEEE Transactions on Automatic Control}, 62(8), 3861--3876.

\bibitem[{Antonelli et~al.(2008)Antonelli, Arrichiello, and Chiaverini}]{Antonelli:2008}
Antonelli, G., Arrichiello, F., and Chiaverini, S. (2008).
\newblock The null-space-based behavioral control for autonomous robotic systems.
\newblock \emph{Intelligent Service Robotics}, 1(1), 27--39.

\bibitem[{Atman et~al.(2019)Atman, Noda, Funada, Yamauchi, Hatanaka, and Fujita}]{Atman:2019}
Atman, M., Noda, K., Funada, R., Yamauchi, J., Hatanaka, T., and Fujita, M. (2019).
\newblock On passivity-shortage of human operators for a class of semi-autonomous robotic swarms.
\newblock \emph{IFAC-PapersOnLine}, 51(34), 21--27.

\bibitem[{Cort{\'e}s et~al.(2005)Cort{\'e}s, Mart{\'i}nez, and Bullo}]{Cortes:2005}
Cort{\'e}s, J., Mart{\'i}nez, S., and Bullo, F. (2005).
\newblock Spatially-distributed coverage optimization and control with limited-range interactions.
\newblock \emph{ESAIM: Control Optimisation and Calculus of Variations}, 11(4), 691--719.

\bibitem[{Dan et~al.(2021)Dan, Hatanaka, Yamauchi, Shimizu, and Fujita}]{Dan:2021}
Dan, H., Hatanaka, T., Yamauchi, J., Shimizu, T., and Fujita, M. (2021).
\newblock Persistent object search and surveillance control with safety certificates for drone networks based on control barrier functions.
\newblock \emph{Frontiers in Robotics and AI}, 8, 740460.

\bibitem[{Diaz-Mercado et~al.(2017)Diaz-Mercado, Lee, and Egerstedt}]{Diaz:2017}
Diaz-Mercado, Y., Lee, S.G., and Egerstedt, M. (2017).
\newblock Human-swarm interactions via coverage of time-varying densities.
\newblock In Y.~Wang and F.~Zhang (eds.), \emph{Trends in Control and Decision-Making for Human-Robot Collaboration Systems}, 357--385. Springer.

\bibitem[{Edmonds and Yi(2021)}]{Edmonds:2021}
Edmonds, M. and Yi, J. (2021).
\newblock Efficient multi-robot inspection of row crops via kernel estimation and region-based task allocation.
\newblock In \emph{2021 IEEE International Conference on Robotics and Automation}, 8919--8926.

\bibitem[{Egerstedt(2021)}]{Magnus:2021}
Egerstedt, M. (2021).
\newblock \emph{Robot ecology: Constraint-based control design for long duration autonomy}.
\newblock Princeton University Press.

\bibitem[{Franchi et~al.(2012{\natexlab{a}})Franchi, Secchi, Ryll, Bulthoff, and Giordano}]{Franchi:2012a}
Franchi, A., Secchi, C., Ryll, M., Bulthoff, H.H., and Giordano, P.R. (2012{\natexlab{a}}).
\newblock Shared control : Balancing autonomy and human assistance with a group of quadrotor {UAV}s.
\newblock \emph{IEEE Robotics \& Automation Magazine}, 19(3), 57--68.

\bibitem[{Franchi et~al.(2012{\natexlab{b}})Franchi, Secchi, Son, Bulthoff, and Giordano}]{Franchi:2012b}
Franchi, A., Secchi, C., Son, H.I., Bulthoff, H.H., and Giordano, P.R. (2012{\natexlab{b}}).
\newblock Bilateral teleoperation of groups of mobile robots with time-varying topology.
\newblock \emph{IEEE Transactions on Robotics}, 28(5), 1019--1033.

\bibitem[{Hanif et~al.(2025)Hanif, Terunuma, Sumino, Cheng, and Hatanaka}]{Hanif:2025}
Hanif, M., Terunuma, R., Sumino, T., Cheng, K., and Hatanaka, T. (2025).
\newblock Coverage-{Recon}: Coordinated multi-drone image sampling with online map feedback.
\newblock \emph{arXiv:2510.18347}.

\bibitem[{Hatanaka et~al.(2024)Hatanaka, Mochizuki, Sumino, Maestre, and Chopra}]{Hatanaka:2024}
Hatanaka, T., Mochizuki, T., Sumino, T., Maestre, J.M., and Chopra, N. (2024).
\newblock Human modeling and passivity analysis for semi-autonomous multi-robot navigation in three dimensions.
\newblock \emph{IEEE Open Journal of Control Systems}, 3, 45--57.

\bibitem[{Hatanaka et~al.(2023)Hatanaka, Yamauchi, Fujita, and Handa}]{Hatanaka:2023}
Hatanaka, T., Yamauchi, J., Fujita, M., and Handa, H. (2023).
\newblock Contemporary issues and advances in human-robot collaborations.
\newblock In A.M. Annaswamy, P.P. Khargonekar, F.~Lamnabhi-Lagarrigue, and S.K. Spurgeon (eds.), \emph{Cyber-Physical-Human Systems: Fundamentals and Applications}, 365--399. Wiley.

\bibitem[{Lee and Spong(2005)}]{Lee:2005}
Lee, D. and Spong, M.W. (2005).
\newblock Bilateral teleoperation of multiple cooperative robots over delayed communication networks: Theory.
\newblock In \emph{Proceedings of the 2005 IEEE International Conference on Robotics and Automation}, 360--365.

\bibitem[{Lu et~al.(2024)Lu, Hanif, Shimizu, and Hatanaka}]{Lu:2024}
Lu, Z., Hanif, M., Shimizu, T., and Hatanaka, T. (2024).
\newblock Angle-aware coverage with camera rotational motion control.
\newblock \emph{SICE Journal of Control, Measurement, and System Integration}, 17(1), 211--221.

\bibitem[{Music et~al.(2017)Music, Salvietti, Dohmann, Chinello, Prattichizzo, and Hirche}]{Music:2017}
Music, S., Salvietti, G., Dohmann, P., Chinello, F., Prattichizzo, D., and Hirche, S. (2017).
\newblock Human-multi-robot teleoperation for cooperative manipulation tasks using wearable haptic devices.
\newblock In \emph{IEEE/RSJ International Conference on Intelligent Robots and Systems}.

\bibitem[{Ott et~al.(2008)Ott, Kugi, and Nakamura}]{Ott:2008}
Ott, C., Kugi, A., and Nakamura, Y. (2008).
\newblock Resolving the problem of non-integrability of nullspace velocities for compliance control of redundant manipulators by using semidefinite lyapunov functions.
\newblock In \emph{2008 IEEE International Conference on Robotics and Automation}, 1999--2004.

\bibitem[{Palacios-Gas{\'o}s et~al.(2016)Palacios-Gas{\'o}s, Montijano, Sag{\"u}{\'e}s, and Llorente}]{Palacios:2016}
Palacios-Gas{\'o}s, J.M., Montijano, E., Sag{\"u}{\'e}s, C., and Llorente, S. (2016).
\newblock Distributed coverage estimation and control for multirobot persistent tasks.
\newblock \emph{IEEE Transactions on Robotics}, 32(6), 1444--1460.

\bibitem[{Schwager et~al.(2011)Schwager, Julian, Angermann, and Rus}]{Schwager:2011}
Schwager, M., Julian, B.J., Angermann, M., and Rus, D. (2011).
\newblock Eyes in the sky: Decentralized control for the deployment of robotic camera networks.
\newblock \emph{Proceedings of the IEEE}, 99(9), 1541--1561.

\bibitem[{Seraj and Gombolay(2020)}]{Seraj:2020}
Seraj, E. and Gombolay, M. (2020).
\newblock Coordinated control of {UAV}s for human-centered active sensing of wildfires.
\newblock In \emph{2020 American Control Conference}, 1645--1652.

\bibitem[{Shimizu et~al.(2022)Shimizu, Yamashita, Hatanaka, Uto, Mammarella, and Dabbene}]{Shimizu:2022}
Shimizu, T., Yamashita, S., Hatanaka, T., Uto, K., Mammarella, M., and Dabbene, F. (2022).
\newblock Angle-aware coverage control for 3-{D} map reconstruction with drone networks.
\newblock \emph{IEEE Control Systems Letters}, 6, 1831--1836.

\bibitem[{Sun et~al.(2021)Sun, Xie, Chen, Zhou, and Bao}]{Sun:2021}
Sun, J., Xie, Y., Chen, L., Zhou, X., and Bao, H. (2021).
\newblock Neural{Recon}: Real-time coherent 3{D} reconstruction from monocular video.
\newblock In \emph{Proceedings of the IEEE/CVF Conference on Computer Vision and Pattern Recognition}, 15593--15602.

\bibitem[{Torres et~al.(2016)Torres, Pelta, Verdegay, and Torres}]{Torres:2016}
Torres, M., Pelta, D.A., Verdegay, J.L., and Torres, J.C. (2016).
\newblock Coverage path planning with unmanned aerial vehicles for 3{D} terrain reconstruction.
\newblock \emph{Expert Systems with Applications}, 55, 441--451.

\bibitem[{Xiao et~al.(2021)Xiao, Tan, and Wang}]{Xiao:2021}
Xiao, S., Tan, X., and Wang, J. (2021).
\newblock A simulated annealing algorithm and grid map-based {UAV} coverage path planning method for 3{D} reconstruction.
\newblock \emph{Electronics}, 10(7), 853.

\end{thebibliography}
% with bibtex (preferred)

% \appendix
% \section{A summary of Latin grammar}    % Each appendix must have a short title.
% \section{Some Latin vocabulary}              % Sections and subsections are supported
% in the appendices.
\end{document}